\documentclass{lmcs}

\usepackage{amssymb}
\usepackage{float}
\usepackage{booktabs}
\usepackage{graphicx}

\newcommand{\agdaindexurl}{https://arxiv.org/src/2610.05187v2/anc/html/Revisited.index.html}
\newcommand{\agdaicon}{%
  \raisebox{-0.2ex}{%
    \includegraphics[height=2ex]{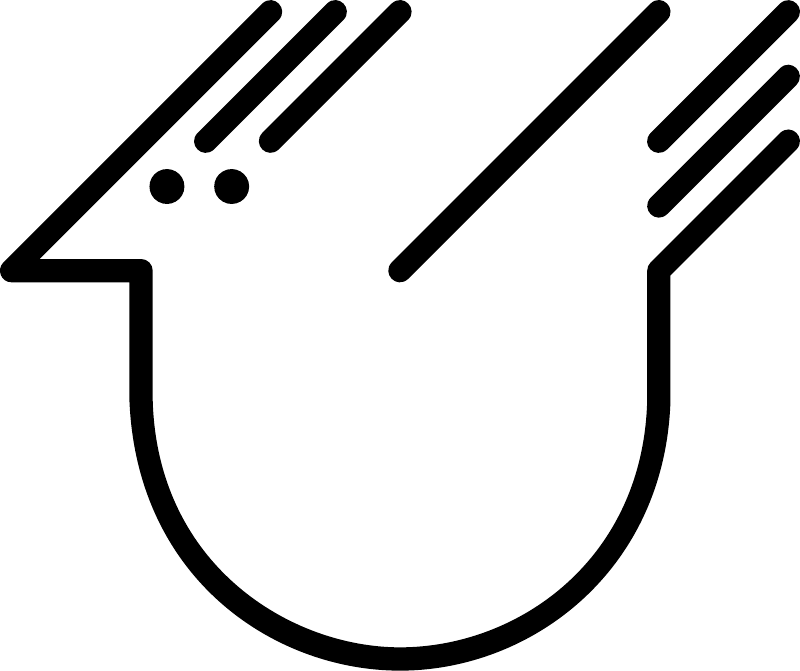}}}
\newcommand{\agdalink}[1]{%
  \href{\agdaindexurl\##1}{{\normalfont\agdaicon}}}

\newcommand{\eqdef}{\coloneqq}
\newcommand{\eg}{\emph{e.g.}}

\newcommand{\IQL}{\mathsf{IQL}}
\newcommand{\CQL}{\mathsf{CQL}}
\newcommand{\MQL}{\mathsf{MQL}}

\newcommand{\HAo}{\mathsf{HA}^\omega}

\newcommand{\T}{\mathsf{T}}
\newcommand{\N}{\mathbb{N}}

\newcommand{\rV}{\mathrm{V}}
\newcommand{\rM}{\mathrm{M}}
\newcommand{\rB}{\mathrm{B}}
\newcommand{\rS}{\mathrm{S}}
\newcommand{\fst}{\mathsf{pr}_1}
\newcommand{\snd}{\mathsf{pr}_2}
\newcommand{\inl}{\mathsf{inj}_1}
\newcommand{\inr}{\mathsf{inj}_2}
\newcommand{\tpair}{\mathsf{Pair}}
\newcommand{\pair}[2]{\left\langle#1,\,#2\right\rangle}
\newcommand{\case}{\mathsf{case}}
\newcommand{\J}{\mathsf{J}}

\newcommand{\R}[1]{\mathrel{\mathrm{R}_{#1}}}

\newcommand{\maj}[1]{\mathrel{\triangleleft_{#1}}}

\newcommand{\E}[2]{\mathrel{\mathrm{R}^{#1}_{#2}}}

\newcommand{\nil}{\mathsf{nil}}

\newcommand{\itp}[1]{[\![#1]\!]}
\newcommand{\Ge}{\mathrm{G}}
\newcommand{\ko}[1]{\langle #1 \rangle}
\newcommand{\ku}[1]{[#1]}
\newcommand{\datatrans}[1]{\langle #1\rangle^{\mathrm D}}

\newcommand{\seq}[1]{(#1_i)_{i \in \N}}
\newcommand{\mcD}{\mathcal{D}}
\newcommand{\leaf}{\eta_\mcD}
\newcommand{\branch}{\beta_\mcD}
\newcommand{\msD}{\mathsf{D}}
\newcommand{\eval}{\mathsf{eval}}

\newcommand{\Ttwo}{(\N \to \N) \to \N}
\newcommand{\Tthree}{\bigl((\N \to \N) \to \N\bigr) \to \N}

\newcommand{\Term}{\mathsf{Term}}
\newcommand{\emptyCtx}{\varepsilon}

\newcommand{\natrec}{\mathsf{natrec}}
\newcommand{\listrec}{\mathsf{listrec}}

\newcommand{\Fst}{\mathsf{Pr1}}
\newcommand{\Snd}{\mathsf{Pr2}}
\newcommand{\Zero}{\mathsf{Zero}}
\newcommand{\Suc}{\mathsf{Suc}}
\newcommand{\Rec}{\mathsf{Rec}}
\newcommand{\KE}{\mathsf{KE}}
\newcommand{\Max}{\mathsf{Max}}
\newcommand{\Up}{\mathsf{Up}}
\newcommand{\Down}{\mathsf{Down}}
\newcommand{\MaxUpTo}{\mathsf{MaxUpTo}}
\newcommand{\RangeBound}{\mathsf{RangeBound}}
\newcommand{\Cons}{\mathsf{Cons}}
\newcommand{\Eval}{\mathsf{Eval}}
\newcommand{\Enc}{\mathrel{\mathrm{Enc}}}
\newcommand{\Append}{\mathsf{Append}}
\newcommand{\Empty}{\mathsf{Empty}}
\newcommand{\Len}{\mathsf{Len}}
\newcommand{\Pad}{\mathsf{Pad}}
\newcommand{\Lt}{\mathsf{Lt}}
\newcommand{\GBR}{\mathsf{GBR}}
\newcommand{\SBR}{\mathsf{SBR}}
\newcommand{\MBR}{\mathsf{MBR}}
\newcommand{\Inl}{\mathsf{Inj1}}
\newcommand{\Inr}{\mathsf{Inj2}}
\newcommand{\List}{\mathsf{List}}
\newcommand{\Case}{\mathsf{Case}}
\newcommand{\Nil}{\mathsf{Nil}}
\newcommand{\LRec}{\mathsf{LRec}}
\newcommand{\SR}[1]{\mathrel{\mathrm{S}_{#1}}}
\newcommand{\App}{\mathsf{Ap}}

\newcommand{\bV}[1]{\mathrel{\mathbf{V}_{#1}}}
\newcommand{\bC}[1]{\mathrel{\mathbf{C}_{#1}}}

\newcommand{\bVV}[1]{\mathrel{\mathbf{V}^{n,\alpha}_{#1}}}
\newcommand{\bCC}[1]{\mathrel{\mathbf{C}^{n,\alpha}_{#1}}}
\newcommand{\bCCz}[1]{\mathrel{\mathbf{C}^{0,\alpha}_{#1}}}

\begin{document}

\title[The Gentzen-style monadic translation]{The Gentzen-style monadic translation\texorpdfstring{\\}{ }of G\"odel's System T revisited}

\titlecomment{This article extends the author's FSCD 2020 paper, \emph{A Gentzen-style monadic translation of G\"odel's System T}~\cite{xu:gentzen2020}.}

\author[C.~Xu]{Chuangjie Xu\lmcsorcid{0000-0001-6838-4221}}


\begin{abstract}
We revisit the Gentzen-style monadic translation of G\"odel's System~$\T$.
The translation is parametrized by a nucleus, a monad-like structure that need not satisfy the monad laws.
A fundamental theorem of logical relations provides a uniform correctness argument for its instances.
By choosing suitable nuclei, we obtain majorants, moduli of pointwise continuity on Baire space and of uniform continuity on its pointwise bounded subsets, Church encodings of dialogue trees, and general bar recursors, all represented by terms of~$\T$.
For the lifting nucleus, we give a uniform construction of generic elements at arbitrary finite types and use its type-2 instance to prove sequential continuity of $\T$-definable functionals of type~3.
We also extend the translation and its fundamental theorem to sums and finite lists.
Lastly, we prove a fundamental theorem of logical relations for the Kuroda-style translation and use it to verify the extracted moduli of pointwise continuity.
\end{abstract}

\maketitle

\section{Introduction}
\label{sec:intro}

G\"odel's System~$\T$ is a language for primitive recursive functionals of finite type.
From its terms, one can extract majorants, moduli of continuity, and functionals witnessing general bar recursion via different methods~\cite{escardo:dialogue,howard:majorizable,kohlenbach:majorization,OS:bar,powell:ccht}.
A common question is how to obtain such witnesses systematically while keeping them definable in System~$\T$.

The Gentzen-style monadic translation addresses this question through a single syntactic construction parametrized by a nucleus.
In its basic form, a \emph{nucleus} consists of a type of System~$\T$ and two families of terms indexed by contexts.
These families supply operations analogous to the unit and Kleisli extension of a monad, but without the requirement that they satisfy the monad laws.
The nucleus determines the translation of the base type and the numerical constants; function and product types are translated componentwise, while abstraction and application retain their form.
Since the nucleus itself is given within~$\T$, the translated terms and the witnesses extracted from them remain in~$\T$.
A fundamental theorem of logical relations provides the basis for verifying these witnesses: if a suitable logical relation is preserved by the operations of the nucleus, then each term and its translation have related interpretations under related assignments.
The central task in each application is therefore to choose a nucleus carrying the desired information and a logical relation expressing its correctness.

The construction was motivated by the syntactic treatment of bar recursion by Oliva and Steila~\cite{OS:bar} and its adaptation to continuity~\cite{xu:tcont}.
Its proof-theoretic inspiration comes from generalizations of negative translations in which double negation is replaced by a nucleus~\cite{EO:peirce,ishihara:GG,vdb:kuroda}.
Under the proofs-as-programs correspondence, our translation follows the pattern of Gentzen's negative translation.
The resulting framework gives a uniform account of majorizability, pointwise and uniform continuity, general bar recursion, and internal dialogue trees.
In each case, the choice of nucleus determines the computational information obtained, while the fundamental theorem provides a common correctness argument.

\paragraph{Contributions of this version}

This article extends our FSCD 2020 paper~\cite{xu:gentzen2020}.
The main additions are as follows:
\begin{itemize}
\item We construct generic elements at arbitrary finite types for the lifting nucleus and prove their genericity under function extensionality.
This gives a uniform method for transferring predicates, which we use to prove sequential continuity of $\T$-definable functionals of type~3 under function extensionality (Section~\ref{sec:lift}).
\item We recover the internal dialogue-tree construction of~\cite{escardo:internal} as an instance of our monadic translation.
Given a System~$\T$ term representing a type-2 functional, we extract a term that encodes a dialogue tree computing that functional.
We give a direct correctness proof using the fundamental theorem with an inductively defined encoding relation (Section~\ref{sec:dialogue}).
\item We extend the translation to finite lists and prove a fundamental theorem of logical relations covering both sums and lists.
This also supplies the correctness theorem left open for sums in the conference version (Section~\ref{sec:gen}).
\item We prove a fundamental theorem of logical relations for the Kuroda-style translation and instantiate it with a continuity nucleus to verify the extracted moduli of pointwise continuity.
This supplies the correctness argument for the comparison with the Gentzen-style construction (Section~\ref{sec:kuroda}).
\end{itemize}

\paragraph{Organization.}
Section~\ref{sec:background} recalls the proof-theoretic motivation and presents the syntax and standard interpretation of System~$\T$.
Section~\ref{sec:trans} defines the Gentzen-style translation and proves its fundamental theorem, and Section~\ref{sec:app} develops the applications.
Section~\ref{sec:gen} introduces generalized nuclei and extends the Gentzen-style translation and its fundamental theorem to sums and finite lists.
Section~\ref{sec:variants} presents the Kuroda-style translation and its continuity application, followed by a brief account of the Kolmogorov-style translation and a comparison of the three styles.
Section~\ref{sec:conclusion} concludes with directions for future work.

\paragraph{Metatheory}

We work in intensional Martin-L\"of type theory with dependent function and pair types, a universe of small types, and inductive types and families.
In metatheoretic statements, we use logical notation according to the Curry--Howard correspondence.
In particular, $\land$ denotes a product, $\forall$ a dependent function type, and $\exists$ a dependent pair type.
Our object language is G\"odel's System~$\T$.
We write $\itp{-}$ for its standard interpretation in the metatheory and state the correctness of the translated terms using this interpretation.
Where a result uses function extensionality, we state it as a hypothesis.

\paragraph{Formalization}
We have formalized the results in Agda~\cite{agda:software}.
The Agda logo~\href{\agdaindexurl}{\agdaicon} links each marked item to the corresponding declaration in the formalization.
For readability, the paper uses named System~$\T$ variables.
The Agda development instead uses de Bruijn indices.
The source code of the Agda formalization is available at \url{https://github.com/cj-xu/GentzenTrans}.


\section{Background}
\label{sec:background}

We recall the generalization of Gentzen's negative translation that motivates our construction.
We then present the syntax and standard interpretation of G\"odel's System~$\T$.

\subsection{Proof-theoretic translations}
\label{sec:neg-trans}

Gentzen's translation\footnote{Nowadays it is known as the G\"odel-Gentzen negative translation. However, G\"odel and Gentzen treat implication differently, and their translations can be distinguished in affine logic~\cite{AO:affine}.} directly places a double negation in front of atomic formulas, disjunctions and existential quantifiers~\cite{TS:book}.
One can replace double negation with a \emph{nucleus} relative to a logic $\mathsf{L}$: an endofunction $j$ on formulas that introduces no new free variables and for which, for any formulas $A,B$, the following statements are provable in $\mathsf{L}$:
\[
A \to jA
\qquad
(A \to jB) \to jA \to jB
\qquad
(jA)[t/x] \leftrightarrow j(A[t/x]).
\]
Nuclei are also known as \emph{lax modalities}~\cite{aczel:modality} and \emph{strong monads}~\cite{EO:peirce}.
Here we adopt the terminology and definition from~\cite{vdb:kuroda}, which provided the technical motivation for this work.
If~$j$ is a nucleus relative to minimal predicate logic $\MQL$, it determines a proof-theoretic translation of $\MQL$ into itself, consisting of a formula translation $A \mapsto A^\Ge_j$ defined as follows:
\[
 \begin{aligned}
 (A \to B)^\Ge_j & \eqdef A^\Ge_j \to B^\Ge_j &&\quad&
 P^\Ge_j & \eqdef j P \qquad \text{for atomic $P$, incl.~$\bot$} \\
 (A \wedge B)^\Ge_j & \eqdef A^\Ge_j \wedge B^\Ge_j &&&
 (A \vee B)^\Ge_j & \eqdef j (A^\Ge_j \vee B^\Ge_j) \\
 (\forall x A)^\Ge_j & \eqdef \forall x A^\Ge_j &&&
 (\exists x A)^\Ge_j & \eqdef j \exists x A^\Ge_j
 \end{aligned}
\]
and a soundness theorem stating that $\MQL \vdash A$ implies $\MQL \vdash A^\Ge_j$.
Particular choices of nucleus also give sound translations between different logical systems:
\begin{itemize}
\setlength\itemsep{1pt}
\item If $jA= (A \to \bot) \to \bot$, then $\CQL\vdash A$ implies $\MQL \vdash A^\Ge_j$.
\item If $jA=(A \to R) \to R$ for a fixed sentence $R$, then $\CQL\vdash A$ implies $\IQL \vdash A^\Ge_j$.
\item If $jA=A\vee \bot$, then $\IQL\vdash A$ implies $\MQL \vdash A^\Ge_j$.
\end{itemize}
Here $\CQL$ stands for classical predicate logic and $\IQL$ for intuitionistic predicate logic. These results are well-known (see \eg~\cite{ishihara:GG,vdb:kuroda}) and various instances of the translation have been applied in term extraction (see \eg~\cite{EO:peirce,kohlenbach:book}).
These examples illustrate the benefit of parametrizing the translation by a nucleus.
The same formula clauses yield sound translations between different logics as the nucleus varies.
We follow the same approach for G\"odel's System~$\T$, pairing a translation parametrized by a nucleus with a uniform correctness theorem.

Under the proofs-as-programs correspondence, the proof-theoretic translation above motivates our translation of System~$\T$ (Definition~\ref{def:trans}).
Following the clauses for atomic formulas, implication, and conjunction, we map the ground type to the type supplied by the nucleus and translate function and product types componentwise.
The use of the first two nucleus principles, $A \to jA$ and $(A \to jB) \to jA \to jB$, in the soundness proof motivates the operations $\eta$ and $\kappa$, respectively, in our term translation.

\subsection{G\"odel's System T}
\label{sec:T}

Our object language is G\"odel's System~$\T$.
We add binary products as a notational convenience to package the values and witnesses used in the applications in Section~\ref{sec:app}.
But they are not essential because one can work with tuples of types and terms.
We write $\iota$ for the ground type of System~$\T$ and $\N$ for the metatheoretic type of natural numbers.
We use the same notation for function types, product types, abstraction, and application in both languages.

We present the syntax of System~$\T$ and its standard interpretation into the metatheory.
Our translation is defined by structural recursion on terms, and its correctness is proved by reasoning about their interpretations.
The interpretation supplies the computation needed for this purpose, so we omit a separate account of computation rules for System~$\T$.

\paragraph[Syntax]{\agdalink{syntax-of-T} Syntax}

The \emph{types} of System~$\T$ are given by:
\[
\sigma,\tau \Coloneqq \iota \mid \sigma\to\tau \mid \sigma\times\tau.
\]
A \emph{context}~$\Gamma$ is a finite set of typed variables, with each variable assigned a unique type.
We write $\emptyCtx$ for the empty context, and $\Gamma, x:\sigma$ for the extension of $\Gamma$ by a fresh variable~$x$ of type~$\sigma$.
For each context $\Gamma$ and type $\sigma$, we write $\Term(\Gamma, \sigma)$ for the type of System~$\T$ \emph{terms} of type~$\sigma$ in context~$\Gamma$.
Figure~\ref{fig:T-rules} defines this family inductively.
A term is \emph{closed} if its context is empty.

\begin{figure}[H]
\centering
\small
\begin{minipage}{\textwidth}
\[
\begin{gathered}
\frac{(x:\sigma) \in \Gamma}
     {x: \Term(\Gamma,\sigma)}
\qquad
\frac{t: \Term\bigl((\Gamma,x:\sigma),\tau\bigr)}
     {\lambda x^\sigma .t : \Term(\Gamma, \sigma \to \tau)}
\qquad
\frac{t: \Term(\Gamma, \sigma\to\tau) \qquad u: \Term(\Gamma, \sigma)}
     {t(u):\Term(\Gamma, \tau)}
\\[1ex]
\Zero: \Term(\Gamma, \iota)
\qquad
\Suc: \Term(\Gamma, \iota \to \iota)
\qquad
\Rec_\sigma: \Term\bigl( \Gamma, \sigma\to(\iota \to \sigma \to \sigma) \to \iota \to \sigma \bigr)
\\[1ex]
\tpair_{\sigma,\tau}: \Term(\Gamma, \sigma \to \tau \to \sigma\times\tau)
\qquad
\Fst_{\sigma,\tau}: \Term(\Gamma, \sigma\times\tau\to\sigma)
\qquad
\Snd_{\sigma,\tau}: \Term(\Gamma, \sigma\times\tau\to\tau)
\end{gathered}
\]
\end{minipage}
\caption{Term formation rules for System~$\T$ with products.}
\label{fig:T-rules}
\end{figure}

We omit type annotations on bound variables and type indices on constants when they can be inferred.
For each context~$\Gamma$ and types $\sigma,\tau,\rho$, we define the composition term
\[
\begin{aligned}
\mathsf{Comp} &:\Term\bigl(\Gamma,(\tau\to\rho)\to(\sigma\to\tau)\to\sigma\to\rho\bigr) \\
\mathsf{Comp} &\eqdef\lambda f.\lambda g.\lambda x.f(g(x)).
\end{aligned}
\]
For $f : \Term(\Gamma, \tau \to \rho)$ and $g : \Term(\Gamma, \sigma \to \tau)$, we write $f \circ g$ for $\mathsf{Comp}(f)(g)$.

\paragraph[Standard interpretation]{\agdalink{standard-interpretation-of-T} Standard interpretation}

We interpret System~$\T$ types as types in the metatheory by
\[
\itp{\iota} \eqdef \N,
\qquad
\itp{\sigma\to\tau} \eqdef \itp{\sigma} \to \itp{\tau},
\qquad
\itp{\sigma\times\tau} \eqdef \itp{\sigma} \times \itp{\tau}.
\]
The interpretation $\itp{\Gamma}$ of a context is the type of assignments $\gamma$ that give a value $\gamma(x):\itp{\sigma}$ to each variable $(x:\sigma)\in\Gamma$.
For $a:\itp{\sigma}$ and $\gamma : \itp{\Gamma}$, we write $\gamma[x \mapsto a]: \itp{\Gamma, x:\sigma}$ for the extension of $\gamma$ that assigns $a$ to $x$.

Given $t: \Term(\Gamma,\sigma)$ and $\gamma: \itp{\Gamma}$, we define $\itp{t}_\gamma: \itp{\sigma}$ by induction on $t$ as follows:
\[
\begin{aligned}
\itp{x}_\gamma
  &\eqdef \gamma(x),&&&
\itp{\Zero}_\gamma
  &\eqdef 0,&&&
\itp{\tpair_{\sigma,\tau}}_\gamma
  &\eqdef \lambda a^{\itp{\sigma}}.\lambda b^{\itp{\tau}}.\pair{a}{b}\\
\itp{\lambda x^\sigma.t}_\gamma
  &\eqdef \lambda a^{\itp{\sigma}}.\itp{t}_{\gamma[x\mapsto a]},&&&
\itp{\Suc}_\gamma
  & \eqdef \lambda n .\, n+1,&&&
\itp{\Fst_{\sigma,\tau}}_\gamma
  &\eqdef \fst,\\
\itp{t(u)}_\gamma
  &\eqdef \itp{t}_\gamma(\itp{u}_\gamma),&&&
\itp{\Rec_\sigma}_\gamma
  &\eqdef \natrec_{\itp{\sigma}},&&&
\itp{\Snd_{\sigma,\tau}}_\gamma
  &\eqdef \snd.
\end{aligned}
\]
Here $\pair{\text{-}}{\text{-}}$ is the pairing operation and $\fst,\snd$ are the product projections in the metatheory.
The metatheoretic recursor $\natrec_A : A \to (\N \to A \to A) \to \N \to A$ is given by
\[
\natrec_A(a)(f)(0) = a,
\qquad
\natrec_A(a)(f)(n+1) = f(n)(\natrec_A(a)(f)(n)).
\]
For a closed term $t$, we omit the empty assignment and write $\itp{t}$.
A value $a : \itp{\sigma}$ in the metatheory is $\T$-\emph{definable} if there is a closed term $t: \Term(\emptyCtx, \sigma)$ such that $\itp{t} = a$.

If $t:\Term(\Gamma,\sigma)$ has no free variables, its interpretation is independent of the assignment $\gamma:\itp{\Gamma}$, so we write $\itp{t}$ for $\itp{t}_\gamma$.

\section{The Gentzen-style monadic translation}
\label{sec:trans}

We define a syntactic translation of System~$\T$ motivated by the proof-theoretic translation of Section~\ref{sec:neg-trans}, and then establish its soundness through a fundamental theorem of logical relations.

\subsection{Nuclei and the translation}
\label{sec:nuclei-translation}

The translation is parametrized by a nucleus, a monad-like structure that need not satisfy the monad laws.

\begin{defi}[\agdalink{nucleus} Nucleus]
A \emph{nucleus} relative to System~$\T$ is a triple $(\J\iota,\eta,\kappa)$, where $\J\iota$ is a type of System~$\T$ and $\eta$ and $\kappa$ are families of terms indexed by contexts.
For each context~$\Gamma$, the corresponding instances have types
\[
\eta : \Term(\Gamma, \iota \to \J\iota)
\quad\text{and}\quad
\kappa : \Term(\Gamma, (\iota \to \J\iota) \to \J\iota \to \J\iota).
\]
We leave the context parameter of $\eta$ and $\kappa$ implicit and often write simply~$\J$ for the nucleus.
Renaming the variables of a context renames the corresponding instances of \(\eta\) and \(\kappa\) in the same way.
\end{defi}

Note that a nucleus is \emph{not} an extension of System~$\T$, but instead a structure given by a type and terms of System~$\T$.
Therefore, the translated terms remain in System~$\T$ rather than some monadic metatheory such as in~\cite{powell:ccht,uustalu:monad}.
A simple example is the \emph{identity} nucleus given by
\[
\J\iota = \iota,
\qquad
\eta = \lambda n^\iota . n,
\qquad
\kappa = \lambda g^{\iota \to \iota}. \lambda n^\iota . g(n).
\]
More examples are available in Section~\ref{sec:app}.
The notation $\J\iota$ anticipates the generalized nuclei of Section~\ref{sec:gen-nucleus}, where $\J$ is defined on every type of System~$\T$.

We are now ready to construct a syntactic translation of System~$\T$ into itself.
\begin{defi}[\agdalink{GentzenTranslation.Gentzen-translation}]
\label{def:trans}
Let $(\J\iota,\eta,\kappa)$ be a nucleus.
By structural recursion, we assign to each type~$\rho$ of System~$\T$ a type $\rho^\J$ of the same system:
\[
\iota^\J \eqdef \J\iota,
\qquad
(\sigma \to \tau)^\J \eqdef \sigma^\J \to \tau^\J,
\qquad
(\sigma \times \tau)^\J \eqdef \sigma^\J \times \tau^\J.
\]
We obtain the translated context $\Gamma^\J$ by replacing each declaration $x:\sigma$ in $\Gamma$ with $x:\sigma^\J$.
We translate each term $t:\Term(\Gamma,\rho)$ to a term $t^\J:\Term(\Gamma^\J,\rho^\J)$ by structural recursion on $t$.

\smallskip
\noindent\emph{Simply typed $\lambda$-calculus.}
The translation keeps variable names and preserves the form of abstraction and application:
\[
(x)^\J \eqdef x,
\qquad
(\lambda x^\sigma.\, t)^\J \eqdef \lambda x^{\sigma^\J}.\,  t^\J,
\qquad
\bigl(t(u)\bigr)^\J \eqdef t^\J(u^\J).
\]
\noindent\emph{Product constants.}
Each product constant is translated to the same constant with its type indices translated:
\[
(\tpair_{\sigma,\tau})^\J \eqdef \tpair_{\sigma^\J,\tau^\J},
\qquad
(\Fst_{\sigma,\tau})^\J \eqdef \Fst_{\sigma^\J,\tau^\J},
\qquad
(\Snd_{\sigma,\tau})^\J \eqdef \Snd_{\sigma^\J,\tau^\J}.
\]

\noindent\emph{Numerical constants and recursor.}
The zero and successor constants are translated using the nucleus operations directly:
\[
\Zero^\J \eqdef \eta(\Zero),
\qquad
\Suc^\J \eqdef \kappa(\eta\circ\Suc).
\]
The recursor is translated as follows:
\[
(\Rec_\sigma)^\J \eqdef \lambda x^{\sigma^\J}. \lambda f^{\J\iota \to \sigma^\J \to \sigma^\J} . \, \KE_\sigma \bigl( \Rec_{\sigma^\J} (x) (f \circ \eta) \bigr)
\]
where $\KE_\rho : \Term( \Gamma, (\iota \to \rho^\J) \to \J\iota \to \rho^\J)$ is an extension of~$\kappa : \Term(\Gamma, (\iota \to \J\iota) \to \J\iota \to \J\iota)$ defined for each context~$\Gamma$ by structural recursion on the type~$\rho$:
\[
\begin{aligned}
\KE_\iota & \eqdef \kappa \\
\KE_{\sigma \to \tau} & \eqdef \lambda g^{\iota \to \sigma^\J \to \tau^\J} .\, \lambda a^{\J\iota} .\, \lambda x^{\sigma^\J} .\, \hspace{1pt} \KE_\tau \bigl( \lambda n^\iota. g(n)(x) \bigr)(a) \\
\KE_{\sigma\times\tau} &\eqdef \lambda g^{\iota\to\sigma^\J\times\tau^\J} .\, \lambda a^{\J\iota} .\,
\tpair(\KE_\sigma(\Fst \circ g)(a))(\KE_\tau(\Snd \circ g)(a)).
\end{aligned}
\]
\end{defi}

In the clauses above, each occurrence of $\eta$, $\kappa$, or $\KE_\rho$ denotes the instance for the context in which it occurs.
For example, when translating $\Rec_\sigma$ in context $\Gamma$, both $\eta$ and $\KE_\sigma$ are taken in the extended context $\Gamma^\J, x: \sigma^\J, f: \J\iota \to \sigma^\J \to \sigma^\J$.

Because function types are translated componentwise, the translation of an application is again an ordinary application in System~$\T$.
There is no need to introduce a nonstandard, monadic notion of function application, which plays an essential role in the other monadic translations~\cite{powell:ccht,uustalu:monad} as discussed in Section~\ref{sec:variants}.

The nucleus operator $\eta$ turns $\Zero$ into a term of the desired type $\J\iota$.
For successor, $\eta \circ \Suc$ has type $\iota \to \J\iota$, and $\kappa$ adapts it to accept an argument of type~$\J\iota$.
In the recursor clause, precomposing $f$ with $\eta$ lets the step function accept an index of type~$\iota$.
The resulting term $\Rec_{\sigma^\J} (x) (f \circ \eta)$ has type $\iota \to \sigma^\J$.
The operator $\kappa$ handles the result type $\J\iota$.
So we instead apply its extension $\KE_\sigma$ to this term, giving one of the required type $\J\iota \to \sigma^\J$.

\subsection{The fundamental theorem of logical relations}
\label{sec:logical-relations}

We use logical relations to prove the soundness of the translation.\footnote{We owe the idea of proving a unified theorem of logical relations to Thomas Powell.}
We show that a term and its translation have related interpretations whenever their assignments are related.
The argument is uniform in the nucleus, subject to two conditions on the interpretations of $\eta$ and $\kappa$.

\begin{defi}[\agdalink{GentzenTranslation.LogicalRelation.logical-relation} Logical relation]
\label{def:LR}
Let $(\J\iota,\eta,\kappa)$ be a nucleus.
Given a relation $\R{\iota}$ between $\N$ and $\itp{\J\iota}$, we extend it to a relation $\R{\rho}$ between $\itp{\rho}$ and $\itp{\rho^\J}$ for each System~$\T$ type~$\rho$, by structural recursion on~$\rho$.
At~$\iota$, it is the given relation $\R{\iota}$.
The remaining clauses are
\[
\begin{aligned}
f \R{\sigma\to\tau} g & \ \eqdef \ \forall x^{\itp{\sigma}} \, a^{\itp{\sigma^\J}} .\, \bigl( x \R{\sigma} a \to f(x) \R{\tau} g(a) \bigr),\\
u \R{\sigma\times\tau} v & \ \eqdef \ \bigl(\fst(u) \R{\sigma} \fst(v)\bigr) \land \bigl(\snd(u) \R{\tau} \snd(v)\bigr).
\end{aligned}
\]
Assignments $\gamma: \itp{\Gamma}$ and $\delta: \itp{\Gamma^\J}$ are related pointwise:
\[
\gamma \R{\Gamma} \delta
\ \eqdef \
\forall (x:\sigma)\in\Gamma .\ \gamma(x) \R{\sigma} \delta(x).
\]
\end{defi}

Fix a nucleus and a base relation as above.
The base relation expresses the property we want to track through the translation.
To relate the interpretations of a term and its translation, we require the nucleus operations to respect this property.
The following two compatibility conditions express this requirement for the interpretations of $\eta$ and $\kappa$ in every context~$\Gamma$ and under every assignment $\gamma:\itp{\Gamma}$.
\begin{equation}
\label{eq:LR}
\tag{$\dagger$}
\begin{gathered}
\forall n^\N .\, n \R{\iota} \itp{\eta}_\gamma(n),
\qquad
\forall f^{\N \to \N} g^{\N \to \itp{\J\iota}} . \, \bigl( \forall n^\N .\, f(n) \R{\iota} g(n) \bigr) \to f \R{\iota \to \iota} \itp{\kappa}_\gamma(g).
\end{gathered}
\end{equation}
The following lemma extends the condition on $\kappa$ to $\KE_\rho$ for every System~$\T$ type~$\rho$.

\begin{lem}[\agdalink{GentzenTranslation.LogicalRelation.preservation-KE}]
\label{lem:KE-LR}
Assume the condition on $\kappa$ in~\eqref{eq:LR}.
For every type~$\rho$, context~$\Gamma$, and assignment $\gamma: \itp{\Gamma}$, if $f: \N \to \itp{\rho}$ and $g : \N \to \itp{\rho^\J}$ satisfy $f(n) \R{\rho} g(n)$ for all $n:\N$, then
\[
f \R{\iota \to \rho} \itp{\KE_\rho}_\gamma (g).
\]
\end{lem}

\begin{proof}
By structural induction on~$\rho$, keeping the context and assignment arbitrary.

At $\rho = \iota$, the claim is exactly the condition on $\kappa$ in~\eqref{eq:LR}.

At $\rho = \sigma \to \tau$, assume $x \R{\sigma} y$ and apply the induction hypothesis at~$\tau$ to $\lambda n.f(n)(x)$ and $\lambda n.g(n)(y)$.

At $\rho = \sigma \times \tau$, apply the induction hypothesis at $\sigma$ to $\fst \circ f$ and $\fst \circ g$, and the one at~$\tau$ to $\snd \circ f$ and $\snd \circ g$.
\end{proof}

\begin{thm}[\agdalink{GentzenTranslation.LogicalRelation.fundamental-theorem} Fundamental Theorem of Logical Relations]
\label{thm:LR}
If both conditions in~\eqref{eq:LR} hold, then for every term $t: \Term(\Gamma, \rho)$ and every pair of assignments $\gamma: \itp{\Gamma}$ and $\delta: \itp{\Gamma^\J}$,
\[
\gamma \R{\Gamma} \delta
\ \to \
\itp{t}_\gamma \R{\rho} \itp{t^\J}_\delta.
\]
In particular, we have $\itp{t} \R{\rho} \itp{t^\J}$ for every closed $t : \Term(\emptyCtx, \rho)$.
\end{thm}

\begin{proof}
By structural induction on $t$, keeping the related assignments $\gamma$ and $\delta$ arbitrary.

The variable case follows from the assumption $\gamma \R{\Gamma} \delta$.

For an abstraction $\lambda x^\sigma. u$, assume $a \R{\sigma} b$.
Then the extended assignments are related, i.e., $\gamma[x \mapsto a] \R{\Gamma, x:\sigma} \delta[x \mapsto b]$.
The induction hypothesis at~$u$ gives $\itp{u}_{\gamma[x \mapsto a]} \R{\tau} \itp{u^\J}_{\delta[x \mapsto b]}$.

For an application, the two induction hypotheses give the desired result.

The pairing and projection cases follow from the definition of the logical relation.

The case of $\Zero$ follows from the condition on~$\eta$ in~\eqref{eq:LR} at~0.

For $\Suc$, the condition on~$\eta$ in~\eqref{eq:LR} gives $n+1 \R{\iota} \itp{\eta}_\delta (n+1)$ for each $n:\N$.
Then the condition on~$\kappa$ relates the interpretations of $\Suc$ and $\Suc^\J$.

For $\Rec_\sigma$, assume $u \R{\sigma} v$ and $f \R{\iota \to \sigma \to \sigma} g$.
Let $\delta'$ extend $\delta$ by assigning $v$ and $g$ to the two outer binders in $(\Rec_\sigma)^\J$.
Let
\[
F = \natrec_{\itp{\sigma}}(u)(f)
\quad\text{and}\quad
G = \natrec_{\itp{\sigma^\J}}(v)(g \circ \itp{\eta}_{\delta'}).
\]
We prove $F(n) \R{\sigma} G(n)$ for each $n:\N$ by induction on $n$.
At zero, this is $u \R{\sigma} v$.
In the successor case, assume $F(n) \R{\sigma} G(n)$.
The condition on~$\eta$ gives $n \R{\iota} \itp{\eta}_{\delta'}(n)$.
Applying $f \R{\iota \to \sigma \to \sigma} g$ to it and to the induction hypothesis $F(n) \R{\sigma} G(n)$ gives
\[
f(n)(F(n)) \R{\sigma} g(\itp{\eta}_{\delta'}(n))(G(n)),
\]
which is $F(n+1) \R{\sigma} G(n+1)$.
Hence $F(n) \R{\sigma} G(n)$ for all $n:\N$, and Lemma~\ref{lem:KE-LR} gives
\[
F \R{\iota\to\sigma} \itp{\KE_\sigma}_{\delta'}(G).
\]
By the interpretation and the definition of the translation, this says
\[
\itp{\Rec_\sigma}_\gamma(u)(f)
\R{\iota\to\sigma}
\itp{(\Rec_\sigma)^\J}_\delta(v)(g).
\]
Therefore, the interpretations of $\Rec_\sigma$ and $(\Rec_\sigma)^\J$ are related.
\end{proof}

\begin{rem}[Verification in finite-type arithmetic]
The fundamental theorem also has a formulation in Heyting arithmetic in all finite types $\HAo$~\cite{TvD:constructivism}, when the base relation $\R{\iota}$ is given by a formula of $\HAo$.
The clauses defining $\R{\rho}$ then give a formula of $\HAo$ for each System~$\T$ type~$\rho$.
We read System~$\T$ terms as terms of finite-type arithmetic, coding products if necessary.
For open terms, we use separate lists of arithmetic variables for the source and translated contexts.
Suppose the arithmetic version of~\eqref{eq:LR} is provable in $\HAo$ for every context.
The structural inductions above yield, for each term $t: \Term(\Gamma,\rho)$, a derivation in $\HAo$ that $t$ is related to $t^\J$ by $\R{\rho}$ whenever the corresponding variables in their contexts are related.
\end{rem}

\section{Applications of the monadic translation}
\label{sec:app}
We now apply the above framework to reveal various properties and structures of $\T$-definable functions.
In each application, we choose a nucleus and use the fundamental theorem of logical relations to establish a property or verify an extracted witness.

Unless specified otherwise, $\Gamma$ denotes an arbitrary context and $\gamma : \itp{\Gamma}$ an arbitrary assignment.
In each application, the nucleus operations $\eta$ and $\kappa$ are defined in~$\Gamma$ and have no free variables, so we write $\itp{\eta}$ and $\itp{\kappa}$ for their interpretations without assignment subscripts.
We use the same convention for auxiliary terms with no free variables.

\subsection{Majorizability}
\label{sec:maj}
Our first application is to recover Howard's majorizability theorem for System~$\T$~\cite{howard:majorizable}. Majorizability plays an important role in models of higher-order calculi and in the proof mining program~\cite{kohlenbach:book}.

\begin{defi}[\agdalink{Majorizability.majorizability-relation} Howard's majorizability relation]
Howard's majorizability relation $\maj{\rho}$ extends the usual order on natural numbers to all finite types by the function and product clauses of Definition~\ref{def:LR}:
\[
\begin{aligned}
n \maj{\iota} m
  &\eqdef n \leq m,\\
f \maj{\sigma\to\tau} g
  &\eqdef \forall x^{\itp{\sigma}}\,y^{\itp{\sigma}}.\,
    x \maj{\sigma} y \to f(x) \maj{\tau} g(y),\\
u \maj{\sigma\times\tau} v
  &\eqdef \bigl(\fst(u) \maj{\sigma} \fst(v)\bigr)
    \wedge \bigl(\snd(u) \maj{\tau} \snd(v)\bigr).
\end{aligned}
\]
If $a \maj{\rho} b$, we say that $a$ is \emph{majorized} by $b$ and call $b$ a \emph{majorant} of $a$.
\end{defi}

Howard's theorem states that the interpretation of each closed System~$\T$ term has a $\T$-definable majorant.
To compute such majorants using our translation, we first define
\[
\begin{aligned}
\Max &: \Term(\Gamma, \iota \to \iota \to \iota) \\
\Max &\eqdef
\Rec_{\iota\to\iota}(\lambda n.n)
\bigl(\lambda m.\lambda h.\,
  \Rec_\iota(\Suc(m))(\lambda n.\lambda r.\Suc(h(n)))\bigr).
\end{aligned}
\]
Its interpretation is the maximum function and satisfies
\[
\itp{\Max}(0)(n) = n,
\quad
\itp{\Max}(m)(0) = m,
\quad
\itp{\Max}(m+1)(n+1) = \itp{\Max}(m)(n) + 1.
\]

\begin{defi}[\agdalink{Majorizability.majorizability-nucleus} Majorizability nucleus]
\label{def:maj-nucleus}
The \emph{majorizability nucleus} $(\J\iota,\eta,\kappa)$ is given by $\J\iota = \iota$ and
\[
\begin{aligned}
\eta & : \Term(\Gamma,\iota \to \iota) &\qquad&&
\kappa &: \Term(\Gamma, (\iota \to \iota) \to \iota \to \iota) \\
\eta & \eqdef \lambda n^\iota. n, &&&
\kappa & \eqdef
\lambda g^{\iota \to \iota}.\,
  \Rec_\iota\bigl(g(\Zero)\bigr)
    \bigl(\lambda n.\lambda r.\Max(r)(g(\Suc(n)))\bigr).
\end{aligned}
\]
\end{defi}

The interpretation $\itp{\kappa}$ has defining equations
\[
\itp{\kappa}(g)(0) = g(0),
\qquad
\itp{\kappa}(g)(n+1) = \itp{\Max} \bigl(\itp{\kappa}(g)(n)\bigr)\bigl(g(n+1)\bigr).
\]
Thus $\itp{\kappa}(g)(n)$ is the maximum of the values $g(0),g(1),\ldots,g(n)$.

\begin{lem}[\agdalink{Majorizability.maximum-bound}]
\label{lm:maj}
For any $g : \N \to \N$, we have $g(m) \leq \itp{\kappa}(g)(n)$ whenever $m \leq n$.
\end{lem}
\begin{proof}
By induction on $n$.
If $n=0$, then $m=0$ and the claim follows from the first defining equation.
If $m \leq n+1$, there are two cases:
(i)~If $m=n+1$, then the claim follows from the second defining equation.
(ii)~If $m \leq n$, we have
\[
g(m) \leq \itp{\kappa}(g)(n) \leq \itp{\kappa}(g)(n+1)
\]
by the induction hypothesis and the second defining equation.
\end{proof}

\begin{lem}[\agdalink{Majorizability.preservation}]
\label{lem:maj-preservation}
The majorizability nucleus satisfies the compatibility conditions in~\eqref{eq:LR} for the relation~$\maj{\rho}$.
\end{lem}
\begin{proof}
The condition on $\eta$ follows from reflexivity of $\leq$.
For $\kappa$, suppose $f(i)\leq g(i)$ for every $i:\N$.
For $n \leq m$, we have $f(n) \leq g(n)$ by hypothesis and $g(n) \leq \itp{\kappa}(g)(m)$ by Lemma~\ref{lm:maj}.
Hence $f\maj{\iota\to\iota}\itp{\kappa}(g)$.
\end{proof}

\begin{cor}[\agdalink{Majorizability.majorization}]
For each closed term $t : \Term(\emptyCtx,\rho)$, its interpretation is majorized by that of its translation:
$\itp{t} \maj{\rho} \itp{t^\J}.$
\end{cor}
\begin{proof}
Apply Theorem~\ref{thm:LR} with Lemma~\ref{lem:maj-preservation}.
\end{proof}

\begin{rem}
The majorizability nucleus also illustrates how a monad law can fail.
Take $g:\N\to\N$ with $g(0)=1$ and $g(1)=0$.
Then $\itp{\kappa}(g)(\itp{\eta}(1))=1$, whereas $g(1)=0$.
Thus the left identity law fails.
\end{rem}

\subsection{Lifting to higher-order functionals}
\label{sec:lift}

In the majorizability example, we extend a relation on natural numbers to arbitrary finite types using the function and product clauses of Definition~\ref{def:LR}.
However, a given property~$P$ of functions $\itp{X} \to \N$, for a finite type~$X$ of System~$\T$, need not be captured by an extension of this form.
To prove such a property for $\T$-definable functions, we lift the ground type to $X\to\iota$.
This lets $P$ serve directly as the base case of a logical predicate on translated values.

\paragraph{The lifting nucleus}

Fix a type~$X$ of System~$\T$.

\begin{defi}[\agdalink{Lifting.lifting-nucleus} Lifting nucleus]
\label{def:lift-nucleus}
The \emph{lifting nucleus} for~$X$ is the reader nucleus $(\J\iota,\eta,\kappa)$ given by $\J\iota = X\to\iota$ and
\[
\begin{aligned}
\eta & : \Term(\Gamma,\iota\to\J\iota)
&\qquad
\kappa & : \Term\bigl(\Gamma,(\iota\to\J\iota)\to\J\iota\to\J\iota\bigr)\\
\eta & \eqdef \lambda n^\iota.\,\lambda x^X.\,n,
&
\kappa & \eqdef \lambda g^{\iota\to\J\iota}.\,
                 \lambda h^{\J\iota}.\,\lambda x^X.\,g(h(x))(x).
\end{aligned}
\]
\end{defi}

A translated natural number is now interpreted as a function $\itp{X}\to\N$.
Evaluating this function at a given input~$x$ determines the base relation.

\begin{defi}[\agdalink{Lifting.base-lifting-relation} Logical relation for lifting]
\label{def:lift-relation}
For each $x:\itp{X}$, define
\[
n \E{x}{\iota} h \ \eqdef \ n = h(x).
\]
We extend this relation to all finite types by Definition~\ref{def:LR} and write $\gamma \E{x}{\Gamma} \delta$ for the induced pointwise relation on assignments.
\end{defi}

\begin{lem}[\agdalink{Lifting.preservation}]
\label{lem:lift-preservation}
For every $x:\itp{X}$, the lifting nucleus satisfies the compatibility conditions in~\eqref{eq:LR} for the relation $\E{x}{\rho}$.
\end{lem}
\begin{proof}
Fix $x$. The condition on $\eta$ follows from
$\itp{\eta}(n)(x)=n$.
For $\kappa$, suppose $f(m)=g(m)(x)$ for every $m:\N$.
Given $n\E{x}{\iota}h$, we have $n=h(x)$ and thus $\itp{\kappa}(g)(h)(x)=g(h(x))(x)=g(n)(x)=f(n)$.
\end{proof}

\begin{cor}[\agdalink{Lifting.lifting-soundness}]
\label{cor:lift-LR}
For every $x:\itp{X}$ and $t:\Term(\Gamma,\rho)$,
if $\gamma\E{x}{\Gamma}\delta$, then
$\itp{t}_\gamma\E{x}{\rho}\itp{t^\J}_\delta$.
\end{cor}
\begin{proof}
Apply Theorem~\ref{thm:LR} with Lemma~\ref{lem:lift-preservation}.
\end{proof}

\paragraph{Generic elements at finite types}

For the lifting nucleus, a generic element supplies a single argument to the translated function from which we can recover the value of the original function at each input.
We state the definition for an arbitrary nucleus, since we will also use it in later applications.

\begin{defi}[\agdalink{Lifting.generic-element} Generic element]
\label{def:generic-element}
Let $(\J\iota,\eta,\kappa)$ be a nucleus and $X$ a type of System~$\T$.
Suppose that each $x:\itp{X}$ determines a base relation $\E{x}{\iota}$ between $\N$ and $\itp{\J\iota}$, and extend it to all finite types by Definition~\ref{def:LR}.
Borrowing terminology from Escard\'o~\cite{escardo:dialogue}, we call a term $\Omega:\Term(\Gamma,X^\J)$ a \emph{generic element} for this family of relations if
\[
x \E{x}{X} \itp{\Omega}_\gamma
\]
for every assignment $\gamma:\itp{\Gamma}$ and every $x:\itp{X}$.
\end{defi}

For the lifting nucleus and the relation from Definition~\ref{def:lift-relation}, genericity gives the following representation.

\begin{prop}[\agdalink{Lifting.representation}]
\label{lift-representation}
Let $\Omega:\Term(\Gamma^\J,X^\J)$ be a generic element.
For every $f:\Term(\Gamma,X\to\iota)$ and $x:\itp{X}$, if $\gamma\E{x}{\Gamma}\delta$, then
\[
\itp{f}_\gamma(x)=\itp{f^\J(\Omega)}_\delta(x).
\]
\end{prop}
\begin{proof}
Corollary~\ref{cor:lift-LR} gives $\itp{f}_\gamma\E{x}{X\to\iota}\itp{f^\J}_\delta$, and genericity gives $x\E{x}{X}\itp{\Omega}_\delta$.
Applying the former to the latter gives the desired equality.
\end{proof}

\begin{exa}[\agdalink{Lifting.generic-elements-at-types-0-and-1} Generic elements at types 0 and 1]
\label{ex:generic-low-types}
When $X=\iota$, we have $X^\J=\iota\to\iota$ and take
$\Omega_0\eqdef\lambda n^\iota.\,n$.
It is generic because $n=\itp{\Omega_0}(n)$ for every $n:\N$.

When $X=\iota\to\iota$, we take
\[
\begin{aligned}
\Omega_1 &:
  \Term\bigl(\Gamma,((\iota\to\iota)\to\iota)\to(\iota\to\iota)\to\iota\bigr),\\
\Omega_1 &\eqdef
  \lambda g^{(\iota\to\iota)\to\iota}.\,
  \lambda\alpha^{\iota\to\iota}.\,\alpha(g(\alpha)).
\end{aligned}
\]
To check genericity, fix $\alpha:\N\to\N$.
If $n\E{\alpha}{\iota}g$, then $n=g(\alpha)$, so
$\alpha(n)=\alpha(g(\alpha))$.
This is precisely $\alpha(n)\E{\alpha}{\iota}\itp{\Omega_1}(g)$.
\end{exa}

For an arbitrary finite type~$X$, we construct $\Omega_X$ by lifting
the identity function on~$X$.

\begin{defi}[\agdalink{Lifting.construction-of-\%CE\%A9} Construction of $\Omega_X$]
\label{def:generic-lifting}
We define two families of terms
\[
\Up_\rho:\Term\bigl(\Gamma,(X\to\rho)\to\rho^\J\bigr)
\qquad\text{and}\qquad
\Down_\rho:\Term\bigl(\Gamma,\rho^\J\to X\to\rho\bigr)
\]
by simultaneous structural recursion on~$\rho$.
The equations for $\Up$ are
\[
\begin{aligned}
\Up_\iota
  &\eqdef \lambda h.\,h,\\
\Up_{\sigma\to\tau}
  &\eqdef \lambda h.\lambda u.\,
    \Up_\tau\bigl(\lambda x.\,h(x)(\Down_\sigma(u)(x))\bigr),\\
\Up_{\sigma\times\tau}
  &\eqdef \lambda h.\,
    \tpair\bigl(\Up_\sigma(\Fst\circ h)\bigr)
          \bigl(\Up_\tau(\Snd\circ h)\bigr).
\end{aligned}
\]
The equations for $\Down$ are
\[
\begin{aligned}
\Down_\iota
  &\eqdef \lambda h.\,h,\\
\Down_{\sigma\to\tau}
  &\eqdef \lambda h.\lambda x.\lambda a.\,
    \Down_\tau\bigl(h(\Up_\sigma(\lambda y.a))\bigr)(x),\\
\Down_{\sigma\times\tau}
  &\eqdef \lambda h.\lambda x.\,
    \tpair\bigl(\Down_\sigma(\Fst(h))(x)\bigr)
          \bigl(\Down_\tau(\Snd(h))(x)\bigr).
\end{aligned}
\]
We then define
\[
\Omega_X:\Term(\Gamma,X^\J)
\qquad\text{by}\qquad
\Omega_X\eqdef\Up_X(\lambda x^X.x).
\]
\end{defi}

At function types, $\Up_{\sigma\to\tau}$ uses $\Down_\sigma$ to turn its argument $u$ of type $\sigma^\J$ into a function of type $X\to\sigma$.
Conversely, $\Down_{\sigma\to\tau}$ regards its argument $a$ of type $\sigma$ as the constant function $\lambda y.a$ and uses $\Up_\sigma$ to obtain an argument of type $\sigma^\J$.

The two examples above establish genericity directly. For the construction at arbitrary finite types, we use \emph{function extensionality}, the principle that functions agreeing at every argument are equal.

\begin{prop}[\agdalink{Lifting.genericity-of-\%CE\%A9}]
\label{prop:generic-lifting}
Under function extensionality, $\Omega_X$ is generic for each System~$\T$ type~$X$.
\end{prop}
\begin{proof}
We prove simultaneously by induction on $\rho$ that, for arbitrary $x:\itp{X}$,
\[
\forall f^{\itp{X}\to\itp{\rho}} .\, f(x)\E{x}{\rho}\itp{\Up_\rho}(f)
\qquad\text{and}\qquad
\forall u^{\itp{\rho}} \, h^{\itp{\rho^\J}} .\,
\bigl( u\E{x}{\rho}h \to u=\itp{\Down_\rho}(h)(x) \bigr).
\]
Both claims hold immediately at $\iota$.
At product types they follow componentwise from the induction hypotheses.

Suppose $\rho=\sigma\to\tau$.
For the first claim, fix $f: \itp{X} \to \itp{\sigma} \to \itp{\tau}$ and let $a \E{x}{\sigma} b$.
Define
\[
q \eqdef \lambda y^{\itp{X}} .\, f(y)\bigl(\itp{\Down_\sigma}(b)(y)\bigr)
: \itp{X} \to \itp{\tau}.
\]
The second induction hypothesis at~$\sigma$ gives $a = \itp{\Down_\sigma}(b)(x)$.
Hence $f(x)(a) = q(x)$.
Applying the first induction hypothesis at~$\tau$ to~$q$ gives $q(x) \E{x}{\tau} \itp{\Up_\tau}(q)$.
By the definition of $\Up_{\sigma\to\tau}$, this is
\[
f(x)(a) \E{x}{\tau} \itp{\Up_{\sigma\to\tau}}(f)(b)
\]
and thus $f(x) \E{x}{\sigma \to \tau} \itp{\Up_{\sigma\to\tau}}(f)$.

For the second claim, suppose $u \E{x}{\sigma\to\tau} h$ and fix $a:\itp{\sigma}$.
Take $b = \itp{\Up_\sigma}(\lambda y.a)$.
Applying the first induction hypothesis at~$\sigma$ to $\lambda y.a$ gives $a \E{x}{\sigma} b$.
Hence $u(a) \E{x}{\tau} h(b)$, and the second induction hypothesis at~$\tau$ gives
\[
u(a) =
\itp{\Down_\tau}\bigl(h(b)\bigr)(x) =
\itp{\Down_{\sigma\to\tau}}(h)(x)(a).
\]
Function extensionality gives $u=\itp{\Down_{\sigma\to\tau}}(h)(x)$.

Finally, take $\rho=X$ and $f = \lambda y^{\itp{X}} .\, y$ in the first claim.
By the definition of $\Omega_X$, this gives $x \E{x}{X} \itp{\Omega_X}_\gamma$ for every $\gamma:\itp{\Gamma}$ and $x:\itp{X}$.
\end{proof}

\begin{thm}[\agdalink{Lifting.representation-of-closed-terms}]
\label{thm:lift-representation}
Under function extensionality, for every type~$X$ and every closed term $f: \Term(\emptyCtx, X \to \iota)$ of System~$\T$,
\[
\itp{f} = \itp{f^\J(\Omega_X)}
\]
as functions $\itp{X} \to\N$.
\end{thm}

\begin{proof}
Propositions~\ref{lift-representation} and~\ref{prop:generic-lifting} show that $\itp{f}$ and $\itp{f^\J(\Omega_X)}$ agree at every $x:\itp{X}$.
Function extensionality then gives the desired equality of functions.
\end{proof}

\paragraph{Transferring predicates}

To prove a property of a $\T$-definable function, we first establish it for the representation given by Theorem~\ref{thm:lift-representation}.
Since $\itp{\J\iota}=\itp{X}\to\N$, a predicate on such functions can serve as the base case of a logical predicate on translated values.

\begin{defi}[\agdalink{Lifting.Predicates.lifted-predicate} Lifting a predicate]
\label{def:lift-predicate}
Let $P$ be a predicate on functions $\itp{X}\to\N$.
For each type~$\rho$, we define a predicate $Q_\rho$ on
$\itp{\rho^\J}$ by
\[
\begin{aligned}
Q_\iota(f)
  &\eqdef P(f),\\
Q_{\sigma\to\tau}(g)
  &\eqdef \forall a^{\itp{\sigma^\J}}.\,
    \bigl(Q_\sigma(a)\to Q_\tau(g(a))\bigr),\\
Q_{\sigma\times\tau}(w)
  &\eqdef Q_\sigma(\fst(w))\land Q_\tau(\snd(w)).
\end{aligned}
\]
\end{defi}

\begin{prop}[\agdalink{Lifting.Predicates.predicate-transfer-for-T-definable-functionals} Predicate transfer]
\label{lift-transfer}
Let $P$ be a predicate on functions $\itp{X}\to\N$, and let $Q_\rho$ be its lifting as in Definition~\ref{def:lift-predicate}.
Suppose that the following conditions hold.
\begin{enumerate}
\item $P(\itp{\eta}(n))$ holds for every $n:\N$.
\item For every $g:\N\to\itp{X}\to\N$ and $h:\itp{X}\to\N$,
      if $P(g(n))$ holds for every $n:\N$ and $P(h)$ holds,
      then $P(\itp{\kappa}(g)(h))$ holds.
\item $Q_X(\itp{\Omega_X})$ holds.
\end{enumerate}
Then, under function extensionality, every $\T$-definable function
$\itp{X}\to\N$ satisfies~$P$.
\end{prop}

\begin{proof}
Using conditions~(1) and~(2), we prove by structural induction on $t:\Term(\Gamma,\rho)$, as in Theorem~\ref{thm:LR}, that $Q_\rho(\itp{t^\J}_\delta)$ holds whenever $Q_\sigma(\delta(x))$ holds for each $(x:\sigma)\in\Gamma$.

Now let $f:\Term(\emptyCtx,X\to\iota)$ be a closed term.
We obtain $Q_{X\to\iota}(\itp{f^\J})$.
Applying this to $Q_X(\itp{\Omega_X})$ from condition~(3) gives $Q_\iota(\itp{f^\J(\Omega_X)})$, that is, $P(\itp{f^\J(\Omega_X)})$.
By Theorem~\ref{thm:lift-representation}, $\itp{f}=\itp{f^\J(\Omega_X)}$, so $P(\itp{f})$ holds.
\end{proof}

This approach was used in~\cite{xu:tcont} with $X=\iota\to\iota$ to prove continuity of every $\T$-definable functional of type $\Ttwo$.
There, the generic element $\Omega_1$ is used to transfer continuity from the translated representation to the original functional.
In Section~\ref{sec:cont}, we enrich the lifting nucleus with witness components to extract moduli of continuity.

\paragraph{Sequential continuity at type 3}

We now apply the lifting approach to prove sequential continuity of $\T$-definable functionals of type~3.
Since such a functional takes a type-2 functional as input, we first specify what it means for a sequence of these inputs to converge.

\begin{defi}[\agdalink{SequentialContinuity.convergence-and-sequential-continuity} Convergence and sequential continuity]
\label{def:sequential-continuity}
In the following, a property holds \emph{eventually} if it holds
for all $i\geq k$ for some $k:\N$.
\begin{enumerate}
\item A sequence $\seq{\alpha}$ of functions $\N\to\N$ \emph{converges pointwise} to $\alpha$ if, for every $n:\N$, we have $\alpha_i(n)=\alpha(n)$ eventually.
\item A sequence $\seq{f}$ of functionals of type $\Ttwo$ \emph{converges continuously} to $f$ if, whenever $\seq{\alpha}$ converges pointwise to $\alpha$, we have $f_i(\alpha_i)=f(\alpha)$ eventually.
\item A functional $H$ of type $\Tthree$ is \emph{sequentially continuous} if, whenever $\seq{f}$ converges continuously to $f$, we have $H(f_i)=H(f)$ eventually.
\end{enumerate}
\end{defi}

In the Kleene--Kreisel model of continuous functionals, every type-3 functional is sequentially continuous with respect to continuously convergent sequences of \emph{continuous} type-2 inputs \cite[Proposition~2.9(a)]{normann:hierarchies}.
In the full set-theoretic hierarchy, classical logic gives examples of type-3 functionals that fail to be sequentially continuous.
The following result applies to $\T$-definable functionals in that hierarchy, with \emph{arbitrary} type-2 inputs.

\begin{cor}[\agdalink{SequentialContinuity.sequential-continuity}]
\label{sequential-continuity-3}
Under function extensionality, every $\T$-definable functional of type~3 is sequentially continuous.
\end{cor}

\begin{proof}
Take $X = (\iota \to \iota) \to \iota$ and let~$P$ be sequential continuity.
We verify the hypotheses of Proposition~\ref{lift-transfer}.

For every $n:\N$, the function $\itp{\eta}(n)$ is constant and hence satisfies~$P$.

For the condition on $\kappa$, suppose $P(h)$ and $P(g(n))$ for each $n:\N$, and let $\seq{f}$ converge continuously to~$f$.
Sequential continuity of~$h$ gives a~$k$ such that $h(f_i) = h(f)$ for all $i \geq k$.
Sequential continuity of $g(h(f))$ gives an~$l$ such that $g(h(f))(f_i) = g(h(f))(f)$ for all $i \geq l$.
Hence, for every $i\geq\max(k,l)$,
\[
\itp{\kappa}(g)(h)(f_i) =
g(h(f_i))(f_i) =
g(h(f))(f_i) =
g(h(f))(f) =
\itp{\kappa}(g)(h)(f).
\]

It remains to prove $Q_X(\itp{\Omega_X})$.
Fix $u$ satisfying $Q_{\iota\to\iota}(u)$.
We must show that $\itp{\Omega_X}(u)$ is sequentially continuous.
Unfolding the definition of $\Omega_X$ at this type gives
\[
\itp{\Omega_X}(u)(f) = f \bigl(\lambda n^\N .\, u(\itp{\eta}(n))(f)\bigr).
\]
Let $\seq{f}$ converge continuously to $f$.
Because each $\itp{\eta}(n)$ satisfies~$P$, the assumption $Q_{\iota \to \iota}(u)$ gives $P\bigl(u(\itp{\eta}(n))\bigr)$.
Define
$\alpha_i(n) \eqdef u(\itp{\eta}(n))(f_i)$
and
$\alpha(n) \eqdef u(\itp{\eta}(n))(f)$.
For each $n:\N$, we have $\alpha_i(n) = \alpha(n)$ eventually, because $u(\itp{\eta}(n))$ is sequentially continuous.
Hence $\seq{\alpha}$ converges pointwise to $\alpha$.
Continuous convergence of $\seq{f}$ now gives
\[
\itp{\Omega_X}(u)(f_i)=f_i(\alpha_i)=f(\alpha)=\itp{\Omega_X}(u)(f)
\]
eventually.
Thus $\itp{\Omega_X}(u)$ is sequentially continuous, proving $Q_X(\itp{\Omega_X})$.

By Proposition~\ref{lift-transfer}, every $\T$-definable functional of type~3 is sequentially continuous.
\end{proof}

\subsection{Continuity}
\label{sec:cont}

Every $\T$-definable functional of type $\Ttwo$ is pointwise continuous and is uniformly continuous on every compact subspace of Baire space~\cite{beeson:foundations}.
We extract System~$\T$ terms whose interpretations are moduli of pointwise and uniform continuity for such functionals.
For this, we enrich the lifting nucleus of Section~\ref{sec:lift} with a modulus component.
We then verify the extracted moduli using the fundamental theorem of logical relations and suitable generic elements.

Some other approaches establish continuity through the Dialectica interpretation and majorizability~\cite{kohlenbach:majorization}, forcing with generic elements~\cite{cj10}, dialogue trees~\cite{escardo:dialogue}, sheaf models~\cite{escardo:xu:kk,fourman:notions,vdh:moerdijk:sheaf} and stateful computations~\cite{cohen:rahli:continuity}.
In finite-type arithmetic, a syntactic translation similar to ours extracts neighbourhood functions witnessing pointwise continuity~\cite{kawai:neighbourhood}.
A call-by-value monadic translation yields moduli of pointwise continuity for $\T$-definable functionals of type~2~\cite{powell:ccht}; we compare it with ours in Section~\ref{sec:kuroda}.

\subsubsection{Pointwise continuity}
\label{sec:pcont}

We first extract and verify moduli of pointwise continuity for $\T$-definable functionals of type~2.

\begin{defi}[\agdalink{PointwiseContinuity._is-a-modulus-of_at_} Modulus of pointwise continuity]
\label{def:pcont-modulus}
Write $\alpha=_m\beta$ if $\alpha(i)=\beta(i)$ for every $i<m$.
For $f:\Ttwo$ and $\alpha:\N\to\N$, a number $m$ is a
\emph{modulus of continuity} of $f$ at $\alpha$ if
\[
\forall \beta^{\N\to\N}.\,
  \bigl( \alpha=_m\beta \to f(\alpha)=f(\beta) \bigr).
\]
A function $M:\Ttwo$ is a \emph{modulus of continuity} of $f$ if $M(\alpha)$ is such a modulus at every $\alpha$.
\end{defi}

We enrich the lifting nucleus for $X= \iota \to \iota$ with a second component for the modulus.

\begin{defi}[\agdalink{PointwiseContinuity.pointwise-continuity-nucleus} Pointwise-continuity nucleus]
\label{def:pcont-nucleus}
\emergencystretch=1em
The \emph{pointwise-continuity nucleus} $(\J\iota,\eta,\kappa)$ is given by
\[
\begin{aligned}
\J\iota & =
\bigl((\iota\to\iota)\to\iota\bigr)
\times
\bigl((\iota\to\iota)\to\iota\bigr),\\[1ex]
\eta &:\Term(\Gamma,\iota\to\J\iota) \\
\eta &\eqdef \lambda n^\iota .\,
  \tpair(\lambda\alpha.n)(\lambda\alpha.\Zero),\\[1ex]
\kappa&:\Term(\Gamma,(\iota\to\J\iota)\to\J\iota\to\J\iota)\\
\kappa &\eqdef \lambda g^{\iota \to \J\iota} .\, \lambda w^{\J\iota} .\,
  \tpair
    \bigl(\lambda\alpha.\Fst(g(\Fst(w)(\alpha)))(\alpha)\bigr)\\
  &\hspace{99pt}
    \bigl(\lambda\alpha.\,
      \Max\bigl(\Snd(g(\Fst(w)(\alpha)))(\alpha)\bigr)
          \bigl(\Snd(w)(\alpha)\bigr)\bigr).
\end{aligned}
\]
The term $\Max$ has been defined in Section~\ref{sec:maj}.
\end{defi}

For $w:\itp{\J\iota}$, write $\rV_w\eqdef\fst(w)$ and $\rM_w\eqdef\snd(w)$ for its value and modulus components.
The value components in the definitions of $\eta$ and $\kappa$ follow the lifting construction of Section~\ref{sec:lift}.
Because the value component of $\itp{\eta}(n)$ is constantly $n$, its modulus component can be constantly zero.
At input $\alpha$, write $i=\rV_w(\alpha)$ for the index selected by~$w$.
The value of $\itp{\kappa}(g)(w)$ at $\alpha$ is then $\rV_{g(i)}(\alpha)$.
Its modulus must account both for the choice of $i$ and for the value of $\rV_{g(i)}$ at $\alpha$, so we take the maximum of $\rM_w(\alpha)$ and $\rM_{g(i)}(\alpha)$.

\begin{defi}[\agdalink{PointwiseContinuity.base-continuity-relation} Logical relation for pointwise continuity]
\label{def:pcont-relation}
For each $\alpha:\N\to\N$, define the base relation between $n:\N$ and $w:\itp{\J\iota}$ by
\[
n \E{\alpha}{\iota} w
\ \eqdef \
n=\rV_w(\alpha)
\ \land\
\bigl(\forall\beta^{\N\to\N}.\,
  \alpha=_{\rM_w(\alpha)}\beta \to
  \rV_w(\alpha)=\rV_w(\beta)\bigr),
\]
and extend it by Definition~\ref{def:LR}.
\end{defi}

The base relation~$\E{\alpha}{\iota}$ relates $n$ to $w$
when $w$ has value~$n$ at~$\alpha$ and its second component gives a modulus of continuity for $\rV_w$ at $\alpha$.

\begin{lem}[\agdalink{PointwiseContinuity.preservation}]
\label{lem:pcont-preservation}
For every $\alpha:\N\to\N$, the pointwise-continuity nucleus satisfies the compatibility conditions in~\eqref{eq:LR} for the relation $\E{\alpha}{\rho}$.
\end{lem}
\begin{proof}
Fix $\alpha$.
The condition on $\eta$ in~\eqref{eq:LR} follows because, for each $n$, the value component of $\itp{\eta}(n)$ is constantly $n$ and the modulus component is constantly zero.
For $\kappa$, suppose that $f(m) \E{\alpha}{\iota} g(m)$ for each $m:\N$ and $n \E{\alpha}{\iota} w$.
We have $n = \rV_w(\alpha)$ and thus
\[
f(n) = \rV_{g(n)}(\alpha) = \rV_{\itp{\kappa}(g)(w)}(\alpha).
\]
If $\alpha=_{\rM_{\itp{\kappa}(g)(w)}(\alpha)}\beta$, the modulus of~$w$ gives $\rV_w(\beta)=n$, and the modulus of~$g(n)$ gives $\rV_{g(n)}(\beta)=\rV_{g(n)}(\alpha)$.
Hence $f(n) \E{\alpha}{\iota} \itp{\kappa}(g)(w)$, as required.
\end{proof}

At a fixed index $n$, evaluation $\beta\mapsto\beta(n)$ has constant modulus $n+1$.
We pair these in a term $A$ and use $\kappa$ to allow the index itself to depend on the input sequence.

\begin{defi}[\agdalink{PointwiseContinuity.generic-element} Generic element for pointwise continuity]
\label{def:pcont-omega}
Define $\Omega:\Term(\Gamma,\J\iota\to\J\iota)$ by
$\Omega\eqdef\kappa(A)$, where
\[
\begin{aligned}
A &:\Term(\Gamma,\iota\to\J\iota) \\
A &\eqdef \lambda n.\,
  \tpair(\lambda\beta.\beta(n))(\lambda\beta.\Suc(n)).
\end{aligned}
\]
\end{defi}

\begin{lem}[\agdalink{PointwiseContinuity.genericity-of-\%CE\%A9}]
\label{lem:pcont-genericity}
The term $\Omega$ is generic for the family of relations $\E{\alpha}{\rho}$.
\end{lem}
\begin{proof}
For each $\alpha:\N\to\N$ and $n:\N$, $\alpha(n)\E{\alpha}{\iota}\itp{A}(n)$ because agreement on the first $n+1$ entries implies agreement at $n$.
The condition on $\kappa$ now gives $\alpha\E{\alpha}{\iota\to\iota}\itp{\Omega}$.
\end{proof}

\begin{thm}[\agdalink{PointwiseContinuity.pointwise-continuity}]
\label{thm:cont}
For every closed term $f: \Term(\emptyCtx, (\iota \to \iota) \to \iota)$ of System~$\T$, the interpretation of $\Snd(f^\J(\Omega))$ is a modulus of continuity of $\itp{f}$.
\end{thm}

\begin{proof}
Fix $\alpha:\N\to\N$.
With Lemma~\ref{lem:pcont-preservation}, Theorem~\ref{thm:LR} gives $\itp{f}\E{\alpha}{(\iota\to\iota)\to\iota}\itp{f^\J}$.
Applying this to the genericity of $\Omega$
(Lemma~\ref{lem:pcont-genericity}) gives
$\itp{f}(\alpha)\E{\alpha}{\iota} \itp{f^\J(\Omega)}$.
Let $w = \itp{f^\J(\Omega)}$.
Since $\alpha$ was arbitrary, the first clause of the base relation gives pointwise equality of $\itp{f}$ and $\rV_w$, and the second gives that $\rM_w$ is a modulus of continuity of $\rV_w$.
Hence $\rM_w=\itp{\Snd(f^\J(\Omega))}$ is a modulus of continuity of $\itp{f}$.
\end{proof}

\subsubsection{Uniform continuity}
\label{sec:ucont}

We now extract a modulus that works uniformly on all inputs bounded by a given sequence.

\begin{defi}[\agdalink{UniformContinuity._is-a-modulus-of_on-inputs-bounded-by_} Modulus of uniform continuity]
\label{def:ucont-modulus}
Write $\alpha\leq^1\delta$ if $\alpha(i)\leq\delta(i)$ for every $i:\N$.
For $f:\Ttwo$ and $\delta:\N\to\N$, a number $m$ is a \emph{modulus of uniform continuity} of $f$ on the inputs bounded by $\delta$ if
\[
\forall \alpha^{\N\to\N}\beta^{\N\to\N}.\,
\bigl(\alpha\leq^1\delta \land \beta\leq^1\delta
  \land \alpha=_m\beta \to f(\alpha)=f(\beta)\bigr).
\]
A function $M:\Ttwo$ is a \emph{modulus of uniform continuity}
of $f$ if $M(\delta)$ is such a modulus for every $\delta$.
\end{defi}

In the pointwise-continuity nucleus, the modulus component of $\kappa$ uses the modulus of the function selected at a given input.
For uniform continuity, we must account for every index that an input bounded by $\delta$ can select.
We therefore first construct a term that bounds the values of a functional on these inputs, given a
modulus for it there.

\begin{defi}[\agdalink{UniformContinuity.auxiliary-terms} Auxiliary terms for uniform continuity]
\label{def:ucont-aux}
Write $i*\alpha$ for the sequence with head $i$ and tail $\alpha$.
The corresponding System~$\T$ term is
\[
\begin{aligned}
\Cons &: \Term(\Gamma, \iota \to (\iota \to \iota) \to \iota \to \iota)\\
\Cons &\eqdef
\lambda i .\, \lambda \alpha .\,
\Rec_\iota(i)(\lambda j .\, \lambda r .\, \alpha(j)).
\end{aligned}
\]
With it, we define
\[
\begin{aligned}
\RangeBound &: \Term(\Gamma,\iota \to ((\iota\to\iota)\to\iota) \to (\iota\to\iota)\to\iota) \\
\RangeBound
& \eqdef \Rec(\lambda f .\, f)\\
& \hspace{33pt}
  \bigl(\lambda m.\lambda h.\lambda f.\lambda \delta.\,
  \MaxUpTo\bigl(\lambda i.\,
    h(\lambda \beta.f(\Cons(i)(\beta)))(\delta \circ \Suc)\bigr)(\delta(\Zero))\bigr),
\end{aligned}
\]
where $\MaxUpTo:\Term(\Gamma,(\iota\to\iota)\to\iota\to\iota)$ is the term $\kappa$ of the majorizability nucleus (Definition~\ref{def:maj-nucleus}).
\end{defi}

Recall that the interpretation of $\MaxUpTo$ satisfies $g(i) \leq \itp{\MaxUpTo}(g)(n)$ for all $i \leq n$.
Write $\delta^+$ for the tail of $\delta$ and thus $\delta^+(j) = \delta(j+1)$ for all $j:\N$.
The interpretation of $\RangeBound$ satisfies
\[
\begin{aligned}
\itp{\RangeBound}(0)(f)(\delta)
  &= f(\delta),\\
\itp{\RangeBound}(m+1)(f)(\delta)
  &= \itp{\MaxUpTo}\bigl(\lambda i.\,
       \itp{\RangeBound}(m)(\lambda\beta.f(i*\beta))(\delta^+)\bigr)
       (\delta(0)).
\end{aligned}
\]

\begin{lem}[\agdalink{UniformContinuity.range-bound}]
\label{lm:mi}
If $m$ is a modulus of uniform continuity of $f:\Ttwo$ on the
inputs bounded by $\delta$, then
$f(\alpha)\leq\itp{\RangeBound}(m)(f)(\delta)$
for every $\alpha\leq^1\delta$.
\end{lem}

\begin{proof}
By induction on $m$, keeping $f$ and $\delta$ arbitrary.

If $m=0$, we have $f(\alpha) = f(\delta)$ for all $\alpha\leq^1\delta$.

Suppose $m+1$ is a modulus for $f$ on inputs bounded by $\delta$.
For each $i\leq\delta(0)$, the function $\lambda\beta.f(i*\beta)$ has modulus $m$ on the inputs bounded by $\delta^+$.
If $\alpha \leq^1 \delta$, then $\alpha(0) \leq \delta(0)$ and $\alpha^+ \leq^1 \delta^+$.
The induction hypothesis gives
\[
f(\alpha(0) * \alpha^+) \leq
\itp{\RangeBound}(m)(\lambda\beta.f(\alpha(0)*\beta))(\delta^+).
\]
By the properties of $\MaxUpTo$ and $\RangeBound$, we have
\[
\itp{\RangeBound}(m)(\lambda\beta.f(\alpha(0)*\beta))(\delta^+) \leq
\itp{\RangeBound}(m+1)(f)(\delta).
\]
Because $\alpha$ and $\alpha(0)*\alpha^+$ are bounded by $\delta$ and agree on the first $m+1$ entries, we have $f(\alpha)=f(\alpha(0)*\alpha^+)$ and thus $f(\alpha) \leq \itp{\RangeBound}(m+1)(f)(\delta)$.
\end{proof}

We are now ready to construct the nucleus.

\begin{defi}[\agdalink{UniformContinuity.uniform-continuity-nucleus} Uniform-continuity nucleus]
\label{def:ucont-nucleus}
The \emph{uniform-continuity nucleus} $(\J\iota,\eta,\kappa)$ is given by
\[
\begin{aligned}
\J\iota & = ((\iota\to\iota)\to\iota) \times ((\iota\to\iota)\to\iota) \\[1ex]
\eta &:\Term(\Gamma,\iota\to\J\iota)\\
\eta &\eqdef \lambda n^\iota .\,
  \tpair(\lambda\alpha.n)(\lambda\alpha.\Zero),\\[1ex]
\kappa &:\Term(\Gamma,(\iota\to\J\iota)\to\J\iota\to\J\iota) \\
\kappa &\eqdef
\lambda g^{\iota \to \J\iota} .\, \lambda w^{\J\iota} .\,
  \tpair
    \bigl(\lambda\alpha.\Fst(g(\Fst(w)(\alpha)))(\alpha)\bigr)
    \bigl(\lambda\delta.\, \Max(b_{g,w,\delta})
          \bigl(\Snd(w)(\delta)\bigr)\bigr),\\
\end{aligned}
\]
where, in the context extended by $g,w,\delta$, we abbreviate
\[
b_{g,w,\delta} \eqdef \MaxUpTo\bigl(\lambda i.\Snd(g(i))(\delta)\bigr) \bigl(\RangeBound(\Snd(w)(\delta))(\Fst(w))(\delta)\bigr).
\]
\end{defi}

For $w:\itp{\J\iota}$, we again write $\rV_w\eqdef\fst(w)$ and $\rM_w\eqdef\snd(w)$.
As in the pointwise case, the value of $\itp{\kappa}(g)(w)$ at $\alpha$ is $\rV_{g(i)}(\alpha)$ for $i=\rV_w(\alpha)$.
Fix $\delta$ and suppose that $\rM_w(\delta)$ and each $\rM_{g(i)}(\delta)$ are uniform moduli for their respective value components on inputs bounded by $\delta$.
By Lemma~\ref{lm:mi}, every index selected by such an input is at most
\[
r=\itp{\RangeBound}(\rM_w(\delta))(\rV_w)(\delta).
\]
The local term $b_{g,w,\delta}$ computes the maximum of $\rM_{g(i)}(\delta)$ over $i\leq r$.
The modulus component of $\itp{\kappa}(g)(w)$ takes the maximum of this number and $\rM_w(\delta)$.
If two inputs bounded by $\delta$ agree on this many entries, the modulus given by $w$ ensures that they select the same index~$i$.
The modulus given by $g(i)$ then ensures that the selected function takes the same value at both inputs.

\begin{defi}[\agdalink{UniformContinuity.base-continuity-relation} Logical relation for uniform continuity]
\label{def:ucont-relation}
For each $\delta:\N\to\N$, define the base relation between $n:\N$ and $w:\itp{\J\iota}$ by
\[
n \E{\delta}{\iota} w \ \eqdef \
n = \rV_w(\delta) \land
\left( \forall \alpha \, \beta .\, \alpha \leq^1 \delta \wedge \beta \leq^1 \delta \wedge \alpha =_{\rM_w(\delta)} \beta \to \rV_w(\alpha) = \rV_w(\beta) \right)
\]
and extend it to all finite types by Definition~\ref{def:LR}.
\end{defi}

The base relation expresses that $n$ is the value of $w$ at $\delta$ and $\rM_w(\delta)$ is a modulus of uniform continuity of $\rV_w$ on inputs bounded by $\delta$.

\begin{lem}[\agdalink{UniformContinuity.preservation}]
\label{lem:ucont-preservation}
For every $\delta:\N\to\N$, the uniform-continuity nucleus satisfies the compatibility conditions in~\eqref{eq:LR} for the relation $\E{\delta}{\rho}$.
\end{lem}
\begin{proof}
Both conditions follow from the construction of $\eta$ and $\kappa$.
\end{proof}

\begin{defi}[\agdalink{UniformContinuity.generic-element} Generic element for uniform continuity]
\label{def:ucont-omega}
Define $\Omega:\Term(\Gamma,\J\iota\to\J\iota)$ by
$\Omega\eqdef\kappa(A)$ with
\[
\begin{aligned}
A &: \Term(\Gamma,\iota\to\J\iota) \\
A &\eqdef \lambda n.\,
  \tpair(\lambda\beta.\beta(n))(\lambda\beta.\Suc(n)).
\end{aligned}
\]
\end{defi}

\begin{lem}[\agdalink{UniformContinuity.genericity-of-\%CE\%A9}]
\label{lem:ucont-genericity}
The term $\Omega$ is generic for the family of relations
$\E{\delta}{\rho}$.
\end{lem}
\begin{proof}
This follows from the definition of $A$ and the condition on $\kappa$
in~\eqref{eq:LR}.
\end{proof}

With a proof similar to that of Theorem~\ref{thm:cont}, we get the following result.

\begin{thm}[\agdalink{UniformContinuity.uniform-continuity}]
\label{thm:ucont}
For every closed term $f:\Term(\emptyCtx,(\iota\to\iota)\to\iota)$ of System~$\T$, the interpretation of $\Snd(f^\J(\Omega))$ is a modulus of uniform continuity of $\itp{f}$.
\end{thm}

\subsection{Church encodings of dialogue trees}
\label{sec:dialogue}

Escard\'o~\cite{escardo:dialogue} shows that every $\T$-definable functional of type~2 is computed by a dialogue tree and uses this representation to prove continuity.
Sterling~\cite{sterling:brouwer} adapts this dialogue-tree interpretation to prove Brouwer's monotone bar theorem for bars realized by such functionals.

Escard\'o and coauthors~\cite{escardo:internal} extract, from a System~$\T$ term representing a type-2 functional, a Church encoding in System~$\T$ of a dialogue tree computing the same functional.
We recover this extraction as an instance of our Gentzen-style translation.
We give a new correctness proof by applying the fundamental theorem to a logical relation based on an inductively defined encoding relation.

\paragraph{Dialogue trees}

We first recall Escard\'o's dialogue trees~\cite{escardo:dialogue}, instantiated for computing functionals of type $\Ttwo$.

\begin{defi}[\agdalink{DialogueTrees.dialogue-trees-and-evaluation} Dialogue trees and evaluation]
\label{def:dialogue-trees}
The inductive type $\mcD$ of dialogue trees has leaves $\leaf(n)$ for $n:\N$ and query nodes $\branch(\varphi)(i)$ for $\varphi:\N\to\mcD$ and $i:\N$.
Evaluation at an input $\alpha:\N\to\N$ is defined by
\[
\begin{aligned}
\eval_\mcD\bigl(\leaf(n)\bigr)(\alpha) &\eqdef n,\\
\eval_\mcD\bigl(\branch(\varphi)(i)\bigr)(\alpha )&\eqdef \eval_\mcD\bigl(\varphi(\alpha(i))\bigr)(\alpha).
\end{aligned}
\]
\end{defi}

\begin{defi}[\agdalink{DialogueTrees.leaf-replacement} Leaf replacement]
\label{def:dialogue-leaf-replacement}
For $F:\N\to\mcD$, the operation $\kappa_\mcD(F):\mcD\to\mcD$ replaces each leaf labelled $n$ by the tree $F(n)$. It is defined by
\[
\begin{aligned}
\kappa_\mcD(F)\bigl(\leaf(n)\bigr)
  &\eqdef F(n),\\
\kappa_\mcD(F)\bigl(\branch(\varphi)(i)\bigr)
  &\eqdef \branch\bigl(\lambda j.\kappa_\mcD(F)(\varphi(j))\bigr)(i).
\end{aligned}
\]
\end{defi}

Given a tree $d$, we can evaluate $\kappa_\mcD(F)(d)$ at $\alpha$ by first evaluating $d$ at $\alpha$ to obtain an index $n$, then evaluating $F(n)$ at the same input.

\begin{lem}[\agdalink{DialogueTrees.evaluation-of-leaf-replacement}]
\label{lem:dialogue-eval-kappa}
For every $F:\N\to\mcD$, $d:\mcD$, and $\alpha:\N\to\N$,
\[
\eval_\mcD\bigl(\kappa_\mcD(F)(d)\bigr)(\alpha)
  = \eval_\mcD\bigl(F(\eval_\mcD(d)(\alpha))\bigr)(\alpha).
\]
\end{lem}
\begin{proof}
By induction on $d$. At a query node with index $i$, apply the induction hypothesis to the subtree selected by $\alpha(i)$.
\end{proof} 

The $\kappa$ operation of the lifting nucleus in Section~\ref{sec:lift} follows the same pattern.
In both constructions, an intermediate result determines which computation to perform next, and that computation
still has access to the original input.

\paragraph{Church encodings in System~$\T$}

Let $\T_2\eqdef(\iota\to\iota)\to\iota$.

\begin{defi}[\agdalink{DialogueTrees.church-encoding} Church encoding of dialogue trees]
\label{def:dialogue-church}
We encode dialogue trees using the Church type
\[
\msD \eqdef
(\iota\to\T_2) \to
\bigl((\iota\to\T_2)\to\iota\to\T_2\bigr) \to
\T_2.
\]
A term of type $\msD$ takes two arguments specifying operations for leaves and query nodes.
We call them the \emph{leaf algebra} and the \emph{query algebra}, respectively.
The tree constructors and the leaf replacement operation are represented by the following terms:
\[
\begin{aligned}
\eta &:\Term(\Gamma,\iota\to\msD) \\
\eta &\eqdef \lambda n.\lambda e.\lambda q.\,e(n),\\[1ex]
\beta &:\Term(\Gamma,(\iota\to\msD)\to\iota\to\msD) \\
\beta &\eqdef \lambda g.\lambda i.\lambda e.\lambda q.\,
  q(\lambda j.g(j)(e)(q))(i),\\[1ex]
\kappa &:\Term(\Gamma,(\iota\to\msD)\to\msD\to\msD)\\
\kappa &\eqdef \lambda g.\lambda d.\lambda e.\lambda q.\,
  d(\lambda n.g(n)(e)(q))(q).
\end{aligned}
\]
\end{defi}

To evaluate an encoded tree, we choose the leaf and query algebras corresponding to the two clauses of $\eval_\mcD$.
A leaf gives the constant functional with value its label.
At a query node with index~$i$, the input~$\alpha$ selects the branch $\alpha(i)$ and is passed to that branch.

\begin{defi}[\agdalink{DialogueTrees.internal-evaluator} Internal evaluator]
\label{def:dialogue-internal-eval}
The \emph{internal evaluator} is given by the term
\[
\Eval:\Term(\Gamma,\msD\to\T_2)
\qquad
\Eval \eqdef \lambda d .\, d(e_0)(q_0),
\]
with the leaf and query algebras
\[
\begin{aligned}
e_0 &: \Term(\Gamma, \iota \to \T_2) &\qquad&&
q_0 &: \Term(\Gamma, (\iota \to \T_2) \to \iota \to \T_2) \\
e_0 &\eqdef\lambda n.\lambda\alpha.n, &&&
q_0 &\eqdef\lambda g.\lambda i.\lambda\alpha.\,
  g(\alpha(i))(\alpha).
\end{aligned}
\]
\end{defi}

\paragraph{The encoding relation}

A value of the Church type $\itp{\msD}$ can be an arbitrary function of the two algebras.
To express that such a value encodes a dialogue tree, we introduce the following relation.

\begin{defi}[\agdalink{DialogueTrees.encoding-relation}]
\label{def:dialogue-encoding}
The relation $c \Enc d$ between $c:\itp{\msD}$ and $d:\mcD$ is inductively defined by
\[
\frac{}
     {\itp{\eta}(n) \Enc \leaf(n)}
\qquad
\frac{\forall j^\N .\ g(j) \Enc \varphi(j)}
     {\itp{\beta}(g)(i) \Enc \branch(\varphi)(i)},
\]
where $n,i:\N$, $g:\N\to\itp{\msD}$, and $\varphi:\N\to\mcD$.
\end{defi}

The operation $\kappa$ preserves the encoding relation in the following sense.

\begin{lem}[\agdalink{DialogueTrees.encoding-of-leaf-replacement}]
\label{encoding-kappa}
Let $g: \N\to\itp{\msD}$, $F: \N\to\mcD$, $c: \itp{\msD}$,
and $d:\mcD$.
If $g(j)\Enc F(j)$ for every $j:\N$ and $c\Enc d$, then
\[
\itp{\kappa}(g)(c)\Enc\kappa_\mcD(F)(d).
\]
\end{lem}
\begin{proof}
By induction on the evidence of $c\Enc d$.
At a leaf labelled $n$, the claim reduces to $g(n)\Enc F(n)$.
At a query node, apply the query rule to the induction hypotheses for its subtrees.
\end{proof}

The internal evaluator agrees with evaluation of the encoded tree.

\begin{lem}[\agdalink{DialogueTrees.evaluation-of-encoded-tree}]
\label{encoding-eval}
Let $c: \itp{\msD}$ and $d:\mcD$.
If $c\Enc d$, then
\[
\itp{\Eval}(c)(\alpha)=\eval_\mcD(d)(\alpha)
\]
for every $\alpha:\N\to\N$.
\end{lem}

\begin{proof}
By induction on the evidence of $c\Enc d$.
At a leaf, both sides return its label.
At a query node with index $i$, apply the induction hypothesis for the subtree selected by $\alpha(i)$.
\end{proof}

Consequently, two dialogue trees encoded by the same Church value agree at each input.

\begin{cor}[\agdalink{DialogueTrees.evaluation-agrees-for-the-same-encoding}]
\label{eval-agree-on-same-encoding}
Let $c: \itp{\msD}$ and $d,d':\mcD$.
If $c\Enc d$ and $c\Enc d'$, then
\[
\eval_\mcD(d)(\alpha) = \eval_\mcD(d')(\alpha)
\]
for every $\alpha:\N\to\N$.
\end{cor}
\begin{proof}
By Lemma~\ref{encoding-eval}, both evaluations are equal to $\itp{\Eval}(c)(\alpha)$.
\end{proof}

\paragraph{Extraction of internal dialogue trees}

We now use the Church encoding to instantiate the translation.
For a closed term $f$ of type $\T_2$, its translation takes an argument of type $\msD\to\msD$.
We supply a generic element of this type to obtain the extracted term.

\begin{defi}[\agdalink{DialogueTrees.dialogue-tree-nucleus-and-extraction} Dialogue-tree nucleus and extraction]
\label{def:dialogue-extraction}
The \emph{dialogue-tree nucleus} is $(\msD,\eta,\kappa)$, with $\msD$, $\eta$, and $\kappa$ as in Definition~\ref{def:dialogue-church}.
Define
\[
\Omega:\Term(\Gamma,\msD\to\msD)
\quad\text{by}\quad
\Omega\eqdef\kappa(\beta(\eta)).
\]
For $f:\Term(\emptyCtx,\T_2)$, the extracted term is $f^\J(\Omega):\Term(\emptyCtx,\msD)$.
\end{defi}

\begin{defi}[\agdalink{DialogueTrees.base-dialogue-relation} Logical relation for dialogue trees]
\label{def:dialogue-relation}
For each $\alpha:\N\to\N$, define the base relation between $n:\N$ and $c:\itp{\msD}$ by
\[
n\E{\alpha}{\iota}c
\ \eqdef\
\exists d^\mcD .\
  \bigl(c\Enc d\bigr)
  \land
  \bigl(n=\eval_\mcD(d)(\alpha)\bigr),
\]
and extend this to all finite types by Definition~\ref{def:LR}.
\end{defi}

A witness to the base relation supplies a tree encoded by $c$ whose evaluation at $\alpha$ is $n$.

\begin{lem}[\agdalink{DialogueTrees.preservation}]
\label{lem:dialogue-preservation}
For every $\alpha:\N\to\N$, the dialogue-tree nucleus satisfies the compatibility conditions in~\eqref{eq:LR} for the relation $\E{\alpha}{\rho}$.
\end{lem}
\begin{proof}
Fix $\alpha$.
For every $n:\N$, the tree $\leaf(n)$ witnesses $n\E{\alpha}{\iota}\itp{\eta}(n)$.

For $\kappa$, suppose $f(j)\E{\alpha}{\iota}g(j)$ for each $j:\N$ and $n\E{\alpha}{\iota}c$.
Write $F(j)$ and $d$ for the trees supplied by these hypotheses.
By Lemma~\ref{encoding-kappa}, $\itp{\kappa}(g)(c)$ encodes $\kappa_\mcD(F)(d)$.
By Lemma~\ref{lem:dialogue-eval-kappa}, its evaluation is
\[
\eval_\mcD(\kappa_\mcD(F)(d))(\alpha)
 = \eval_\mcD(F(\eval_\mcD(d)(\alpha)))(\alpha)
 = \eval_\mcD(F(n))(\alpha)
 = f(n).
\]
Thus $\kappa_\mcD(F)(d)$ witnesses $f(n)\E{\alpha}{\iota}\itp{\kappa}(g)(c)$.
\end{proof}

\begin{lem}[\agdalink{DialogueTrees.genericity-of-\%CE\%A9}]
\label{lem:dialogue-genericity}
The term $\Omega$ is generic for the family of relations
$\E{\alpha}{\rho}$.
\end{lem}
\begin{proof}
For each $i:\N$, the value $\itp{\beta(\eta)}(i)$ encodes $\branch(\leaf)(i)$, whose evaluation at $\alpha$ is $\alpha(i)$.
The condition on $\kappa$ therefore gives $\alpha\E{\alpha}{\iota\to\iota}\itp{\Omega}$.
\end{proof}

\begin{thm}[\agdalink{DialogueTrees.correctness-of-extraction}]
\label{thm:internal-dialogue}
For every closed term $f:\Term(\emptyCtx, (\iota \to \iota) \to \iota)$ of System~$\T$, there is a dialogue tree $d:\mcD$ such that $\itp{f^\J(\Omega)}\Enc d$ and
\[
\itp{f}(\alpha)=\eval_\mcD(d)(\alpha)
\]
for every $\alpha:\N\to\N$.
\end{thm}

\begin{proof}
Write $c=\itp{f^\J(\Omega)}$.
For each $\alpha:\N\to\N$, Theorem~\ref{thm:LR} and genericity give $\itp{f}(\alpha)\E{\alpha}{\iota}c$.
Take $d$ to be the tree supplied by this relation at the constant zero sequence.
Then $c\Enc d$.
For an arbitrary $\alpha$, the relation supplies a tree $d_\alpha$ with $c\Enc d_\alpha$ and $\itp{f}(\alpha)=\eval_\mcD(d_\alpha)(\alpha)$.
Since $d$ and $d_\alpha$ are encoded by the same value $c$, Corollary~\ref{eval-agree-on-same-encoding} gives
\[
\eval_\mcD(d_\alpha)(\alpha) = \eval_\mcD(d)(\alpha).
\]
Thus $d$ computes $\itp{f}$ at every input.
\end{proof}

The internal evaluator therefore recovers the original functional pointwise.

\begin{cor}[\agdalink{DialogueTrees.evaluation-of-extracted-tree}]
For every closed term $f:\Term(\emptyCtx, (\iota \to \iota) \to \iota)$ and every $\alpha:\N\to\N$,
\[
\itp{f}(\alpha) = \itp{\Eval(f^\J(\Omega))}(\alpha).
\]
\end{cor}

\begin{proof}
By Theorem~\ref{thm:internal-dialogue} and Lemma~\ref{encoding-eval}.
\end{proof}

\subsection{Bar recursion}
\label{sec:bar}
To prove Schwichtenberg's theorem~\cite{schwichtenberg:brct} that the $\T$-definable functionals are closed under a rule-like version of Spector's bar recursion of type levels 0 and 1, Oliva and Steila~\cite{OS:bar} introduce a notion of general bar recursion whose termination condition is given by a decidable monotone bar on finite sequences.
We recover their construction of general bar-recursion functionals~\cite[Definitions~3.1 \& 3.3]{OS:bar} as an instance of our translation.
Here we consider sequences of natural numbers and an arbitrary System~$\T$ result type $\sigma$.

\paragraph{Finite sequences and decidable predicates}

We represent finite sequences of natural numbers by pairs of an infinite sequence and a length, using the metatheoretic type
\[
\N^*\eqdef(\N\to\N)\times\N.
\]
A pair $s=(\alpha,n)$ represents the first $n$ entries of $\alpha$.
We write $|s|\eqdef n$ and define its zero extension by
\[
\hat{s}(i)\eqdef
\begin{cases}
\alpha(i), & i<n,\\
0, & i\geq n.
\end{cases}
\]
We say that $\beta:\N\to\N$ \emph{extends} $s$ if
\[
\forall i^\N.\,i<|s|\to\beta(i)=\hat{s}(i).
\]
For $n:\N$, appending $n$ to $s$ gives
$s*n\eqdef(\beta,|s|+1)$, where
\[
\beta(i)\eqdef
\begin{cases}
\hat{s}(i), & i<|s|,\\
n, & i\geq|s|.
\end{cases}
\]
We take $\nil\eqdef(\lambda i.i,0)$ as the empty sequence.

We represent decidable predicates on finite sequences by functions $S: \N^*\to\N$, with nonzero values representing truth.
We write $S(s)$ for $S(s) \ne 0$; its negation is equivalent to $S(s)=0$.

For the syntactic constructions in System~$\T$, we use the type
\[
\iota^*\eqdef(\iota\to\iota)\times\iota,
\]
whose interpretation is $\N^*$.
We fix terms
\[
\begin{aligned}
\Len&:\Term(\Gamma,\iota^*\to\iota),&&&
\Pad&:\Term(\Gamma,\iota^*\to(\iota\to\iota)),\\
\Append&:\Term(\Gamma,\iota^*\to\iota\to\iota^*),&&&
\Empty&:\Term(\Gamma,\iota^*),
\end{aligned}
\]
with no free variables, whose interpretations satisfy
\[
\begin{aligned}
\itp{\Len}(s)&=|s|,&&&
\itp{\Pad}(s)&=\hat{s},\\
\itp{\Append}(s)(n)&=s*n,&&&
\itp{\Empty}&=\nil.
\end{aligned}
\]
We omit their routine constructions.
A term of type $\iota^*\to\iota$ represents a decidable predicate through its interpretation.

\paragraph{General bar recursion}

Fix a System~$\T$ type $\sigma$.

\begin{defi}[\agdalink{BarRecursion.General.general-bar-recursion-monotonicity-and-security} General bar recursion, monotonicity, and security]
\label{def:gbr}
Let $S:\N^*\to\N$ be a decidable predicate.
A functional
\[
\xi:
(\N^*\to\itp{\sigma})\to
(\N^*\to(\N\to\itp{\sigma})\to\itp{\sigma})\to
\N^*\to\itp{\sigma}
\]
is a \emph{general bar recursor} for $S$ if, for each base functional $G$, step functional $H$, and sequence $s:\N^*$,
\[
\xi(G)(H)(s)=
\begin{cases}
G(s), & \text{if }S(s),\\
H(s)(\lambda n.\xi(G)(H)(s*n)), & \text{otherwise}.
\end{cases}
\]
We denote this property by $\mathcal{GBR}_S(\xi)$.

The predicate~$S$ is \emph{monotone} if $S(s)$ implies $S(s*n)$ for every $s:\N^*$ and $n:\N$.
It \emph{secures} a functional $Y: \Ttwo$ if, whenever $S(s)$ holds, for every $\alpha:\N\to\N$ extending $s$,
\[
Y(\hat{s}) = Y(\alpha).
\]
\end{defi}

For each closed term $Y: \Term(\emptyCtx, (\iota \to \iota) \to \iota)$, Oliva and Steila construct a decidable monotone bar~$S$ securing~$\itp{Y}$ and a general bar recursor for~$S$~\cite[Theorem~3.4]{OS:bar}.
They also give a $\T$-definable conversion which, for any such bar~$S$, takes a general bar recursor for $S$ to a Spector bar recursor for~$\itp{Y}$~\cite[Theorem~2.4]{OS:bar}.
Together, these constructions give an explicit proof of Schwichtenberg's bar-recursion closure theorem for sequences of natural numbers.

To recover the first construction, we use a nucleus whose elements carry a value functional, a predicate on finite sequences, and a recursor.

\begin{defi}[\agdalink{BarRecursion.General.general-bar-recursion-nucleus} Nucleus for general bar recursion]
\label{def:gbr-nucleus}
The \emph{nucleus for general bar recursion} is the triple $(\J\iota,\eta,\kappa)$ with
\[
\J\iota =
\bigl( (\iota \to \iota) \to \iota \bigr) \times
(\iota^* \to \iota) \times
\bigl((\iota^* \to \sigma) \to (\iota^* \to (\iota\to\sigma) \to \sigma) \to \iota^* \to \sigma \bigr).
\]
For a term $w:\Term(\Gamma,\J\iota)$, we abbreviate its three projections by
\[
\mathsf{V}(w)\eqdef\Fst(w),\qquad
\mathsf{S}(w)\eqdef\Fst(\Snd(w)),\qquad
\mathsf{B}(w)\eqdef\Snd(\Snd(w)).
\]
For terms $a,b,c$ of the respective component types, we write $\langle a,b,c\rangle$ for $\tpair(a)(\tpair(b)(c))$.
The nucleus operations are
\[
\begin{aligned}
\eta &:\Term(\Gamma,\iota\to\J\iota) \\
\eta &\eqdef\lambda n.\,
\langle
\lambda\alpha.n,\,
\lambda s.\Suc(\Zero),\,
\lambda G.\lambda H.G
\rangle, \\[1ex]
\kappa &:\Term(\Gamma,(\iota\to\J\iota)\to\J\iota\to\J\iota) \\
\kappa & \eqdef\lambda g.\lambda w.\,
\bigl\langle
\lambda\alpha.\mathsf{V}(g(\mathsf{V}(w)(\alpha)))(\alpha),\\
& \hspace{56pt}
\lambda s.\mathsf{Min}(\mathsf{S}(w)(s))
  \bigl(\mathsf{S}(g(\mathsf{V}(w)(\Pad(s))))(s)\bigr),\\
& \hspace{56pt}
\lambda G.\lambda H.\,
\mathsf{B}(w)
  \bigl(\lambda s.\,
    \mathsf{B}(g(\mathsf{V}(w)(\Pad(s))))(G)(H)(s)\bigr)
  (H)
\bigr\rangle,
\end{aligned}
\]
where $\mathsf{Min}$ is the term representing the minimum function, given by
\[
\begin{aligned}
\mathsf{Min} &:\Term(\Gamma,\iota\to\iota\to\iota) \\
\mathsf{Min} &\eqdef
\Rec_{\iota\to\iota}(\lambda n.\Zero)
\bigl(\lambda m.\lambda h.\,
  \Rec_\iota(\Zero)
    (\lambda n.\lambda r.\Suc(h(n)))\bigr).
\end{aligned}
\]
\end{defi}

For $w:\itp{\J\iota}$, we write $\rV_w,\rS_w,\rB_w$ for its three components.
The intuition is that $\rS_w$ is a monotone predicate securing $\rV_w$ and $\rB_w$ is a general bar recursor for $\rS_w$.
To understand $\itp{\kappa}(g)(w)$, suppose that $\rS_w$ secures $\rV_w$.
If $\rS_w(s)$ holds, every $\alpha$ extending $s$ selects the same index $\rV_w(\alpha)=\rV_w(\hat{s})$.
Write $i$ for this index.
The predicate for $\itp{\kappa}(g)(w)$ therefore requires both $\rS_w(s)$ and $\rS_{g(i)}(s)$.
The recursor first uses $\rB_w$ and, when its base case is reached at $s$, continues with $\rB_{g(i)}$ at $s$.

To construct the generic element, we need a predicate securing the function $\lambda\alpha.\alpha(n)$.
The condition $n<|s|$ suffices, since every extension of such an $s$ agrees at index $n$.

\begin{defi}[\agdalink{BarRecursion.General.generic-element} Generic element for general bar recursion]
\label{def:gbr-omega}
Take the term
\[
\Psi:\Term\Bigl(\Gamma,\,
\iota\to
(\iota^*\to\sigma)\to
(\iota^*\to(\iota\to\sigma)\to\sigma)\to
\iota^*\to\sigma\Bigr)
\]
of Oliva and Steila~\cite[Lemma~2.1]{OS:bar}, defined using primitive recursion.

Define $\Omega:\Term(\Gamma,\J\iota\to\J\iota)$ by
$\Omega\eqdef\kappa(A)$ with
\[
\begin{aligned}
A &:\Term(\Gamma,\iota\to\J\iota)\\
A &\eqdef \lambda n.\,
  \langle\lambda\alpha.\alpha(n),\,
         \lambda s.\Lt(n)(\Len(s)),\,
         \Psi(n)\rangle,
\end{aligned}
\]
where strict comparison is represented by
\[
\begin{aligned}
\Lt &:\Term(\Gamma,\iota\to\iota\to\iota)\\
\Lt &\eqdef \Rec_{\iota\to\iota}
  \bigl(\Rec_\iota(\Zero)(\lambda m.\lambda r.\Suc(\Zero))\bigr)
  \bigl(\lambda n.\lambda h.\,
    \Rec_\iota(\Zero)(\lambda m.\lambda r.h(m))\bigr).
\end{aligned}
\]
\end{defi}

The property of $\Psi$ used below is the following form of
\cite[Lemma~2.1]{OS:bar}.

\begin{lem}[\agdalink{BarRecursion.General.bar-recursion-for-length-comparison}]
\label{lem:gbr-psi}
Assume function extensionality. For every $n:\N$, the functional
$\itp{\Psi}(n)$ is a general bar recursor for the predicate
$\lambda s.\itp{\Lt}(n)(|s|)$.
\end{lem}

\begin{defi}[\agdalink{BarRecursion.General.base-bar-recursion-relation} Logical relation for general bar recursion]
\label{def:gbr-relation}
For each $\alpha:\N\to\N$, define the base relation between $n:\N$ and $w:\itp{\J\iota}$ by
\[
n \E{\alpha}{\iota} w
\ \eqdef\ 
n=\rV_w(\alpha) \land
\rS_w\text{ is monotone} \land
\rS_w\text{ secures }\rV_w \land
\mathcal{GBR}_{\rS_w}(\rB_w)
\]
and extend it to all finite types by Definition~\ref{def:LR}.
\end{defi}

\begin{lem}[\agdalink{BarRecursion.General.preservation}]
\label{lem:gbr-preservation}
For every $\alpha:\N\to\N$, the nucleus for general bar recursion satisfies the compatibility condition on $\eta$ in~\eqref{eq:LR} for the relation $\E{\alpha}{\rho}$.
Under function extensionality, it also satisfies the condition on $\kappa$.
\end{lem}
\begin{proof}
The condition on $\eta$ follows directly from its definition.

For $\kappa$, suppose $f(i)\E{\alpha}{\iota}g(i)$ for every $i:\N$ and $n\E{\alpha}{\iota}w$.
Put $v=\itp{\kappa}(g)(w)$.
The value equation follows from $n=\rV_w(\alpha)$ and $f(n)=\rV_{g(n)}(\alpha)$.
Whenever $\rS_w(s)$ holds, security keeps the index $i=\rV_w(\hat{s})$ fixed on every extension of $s$.
Monotonicity and security of $\rS_v$ now follow from those of $\rS_w$ and $\rS_{g(i)}$.

For the bar-recursion equations, if $\rS_v(s)$ holds, apply the base equations for $\rB_w$ and $\rB_{g(i)}$ in turn.
If $\rS_w(s)$ fails, the step equation for $\rB_w$ gives the required step equation directly.
It remains to consider the case where $\rS_w(s)$ holds and $\rS_{g(i)}(s)$ fails.
Here the base equation for $\rB_w$ is followed by the step equation for $\rB_{g(i)}$.
At every $s*m$, monotonicity gives $\rS_w(s*m)$, and security gives $\rV_w(\widehat{s*m})=i$.
Thus the recursive calls agree pointwise with those of $\rB_v$.
Function extensionality gives equality of the continuations supplied to $H(s)$.
\end{proof}

\begin{lem}[\agdalink{BarRecursion.General.genericity-of-\%CE\%A9}]
\label{lem:gbr-genericity}
Under function extensionality, $\Omega$ is generic for the family of relations $\E{\alpha}{\rho}$.
\end{lem}
\begin{proof}
For each $n:\N$, the predicate $n<|s|$ is monotone and secures $\lambda\alpha.\alpha(n)$.
Together with Lemma~\ref{lem:gbr-psi}, this gives $\alpha(n)\E{\alpha}{\iota}\itp{A}(n)$ for every $\alpha:\N\to\N$.
The condition on $\kappa$ now gives $\alpha\E{\alpha}{\iota\to\iota}\itp{\Omega}$.
\end{proof}

\begin{thm}[\agdalink{BarRecursion.General.extraction-of-general-bar-recursors}]
\label{thm:gbr}
Under function extensionality, for each $Y:\Term(\emptyCtx,(\iota \to \iota) \to \iota)$,
\begin{enumerate}
\item $\rS_{\itp{Y^\J(\Omega)}}$ is a monotone predicate securing $\itp{Y}$, and
\item $\rB_{\itp{Y^\J(\Omega)}}$ is a general bar recursor for $\rS_{\itp{Y^\J(\Omega)}}$.
\end{enumerate}
\end{thm}

\begin{proof}
Let $w = \itp{Y^\J(\Omega)}$.
By Theorem~\ref{thm:LR} and Lemmas~\ref{lem:gbr-preservation} and~\ref{lem:gbr-genericity}, we have $\itp{Y}(\alpha)\E{\alpha}{\iota}w$ for every $\alpha:\N\to\N$.
The first clause gives pointwise equality of $\itp{Y}$ and $\rV_w$.
The other clauses give monotonicity of $\rS_w$, security of $\rV_w$, and $\mathcal{GBR}_{\rS_w}(\rB_w)$.
Security transfers to $\itp{Y}$ by pointwise equality.
\end{proof}

\begin{rem}
Our construction follows Oliva and Steila~\cite{OS:bar}, with three differences:
\begin{itemize}
\item Oliva and Steila require~$S$ to satisfy the bar condition $\forall \alpha^{\N \to \N} \exists n^\N S(\bar{\alpha}n)$, where $\bar{\alpha}n : \N^*$ is the prefix of $\alpha$ of length $n$.
As pointed out by Makoto Fujiwara in a personal discussion, this condition can be omitted from our theorem.
Their proof of the conversion to Spector's bar recursion~\cite[Theorem~2.4]{OS:bar} likewise uses only decidability, monotonicity, and security.
\item Oliva and Steila translate terms with a distinguished free variable $\alpha$ of type $\iota\to\iota$, giving a special translation clause for $\alpha$~\cite[Definition~3.3]{OS:bar}.
In our formulation, its role is played by the generic element $\Omega$, which is an ordinary System~$\T$ term.
This was motivated by a version of Escard\'o's Agda development of~\cite{escardo:dialogue}.
\item The syntactic translation of System~$\T$ in~\cite[Definitions 3.1 \& 3.3]{OS:bar} carries the value and bar-recursion components while the monotone securing predicates are constructed in the correctness proof~\cite[Theorem~3.4]{OS:bar}. Our translation constructs all three components to separate their constructions from the correctness proof.
\end{itemize}
\end{rem}

We present three applications of Theorem~\ref{thm:gbr}.
For the remainder of this subsection, assume function extensionality and fix a closed term $Y: \Term(\emptyCtx, (\iota\to\iota)\to\iota)$.
For each choice of result type $\sigma$, write
\[
\GBR_Y\eqdef\mathsf{B}(Y^\J(\Omega)),
\qquad
S_Y\eqdef\rS_{\itp{Y^\J(\Omega)}}.
\]
Thus $\GBR_Y$ is a System~$\T$ term whose interpretation is a general bar recursor for $S_Y$, and $S_Y$ secures $\itp{Y}$.

\paragraph{Uniform continuity}

An anonymous reviewer of the FSCD 2020 version of this paper~\cite{xu:gentzen2020} observed that general bar recursion also yields a modulus of uniform continuity for~$\itp{Y}$ on each of the bounded sets considered in Section~\ref{sec:ucont}.
Take $\sigma = \iota$.
Given a bound $\delta$, the base case returns zero.
At a step from $s$, the possible next entries are the numbers $n\leq\delta(|s|)$, so we take one plus the maximum of their recursive results.

\begin{defi}[\agdalink{BarRecursion.UniformModulus.uniform-modulus} Uniform modulus by general bar recursion]
\label{def:gbr-uniform-modulus}
Define the terms
\[
\begin{aligned}
G & :\Term(\Gamma,\iota^*\to\iota) \\
G & \eqdef \lambda s.\Zero, \\[1ex]
H & :\Term(\Gamma,(\iota\to\iota)\to\iota^*
                    \to(\iota\to\iota)\to\iota) \\
H & \eqdef \lambda\delta.\lambda s.\lambda f.\,
  \Suc(\MaxUpTo(f)(\delta(\Len(s)))).
\end{aligned}
\]
Starting the recursion at the empty sequence gives the closed term
\[
\begin{aligned}
M_Y & :\Term(\emptyCtx,(\iota\to\iota)\to\iota) \\
M_Y & \eqdef
\bigl(\lambda r.\lambda\delta.\, r(G)(H(\delta))(\Empty)\bigr)(\GBR_Y).
\end{aligned}
\]
\end{defi}

\begin{prop}[\agdalink{BarRecursion.UniformModulus.correctness-of-uniform-modulus}]
\label{gbr-uniform-continuity}
For each $\delta: \N \to \N$, $\itp{M_Y}(\delta)$ is a modulus of uniform continuity of~$\itp{Y}$ on inputs bounded by $\delta$.
\end{prop}

\begin{proof}
Fix $\delta:\N\to\N$ and put
\[
B(s) \eqdef \itp{\GBR_Y}(\itp{G})(\itp{H}(\delta))(s).
\]
By induction on $m$, we prove the following stronger claim.
If $\alpha,\beta\leq^1\delta$ both extend $s$, agree on the
next $m$ entries after $s$, and $B(s)\leq m$, then
$\itp{Y}(\alpha)=\itp{Y}(\beta)$.

If $S_Y(s)$ holds, security gives the equality.
Otherwise, the step equation gives $B(s)>0$ and $B(s*n)<B(s)$ for every $n\leq\delta(|s|)$.
Thus $m>0$.
Taking $n=\alpha(|s|)=\beta(|s|)$, we apply the induction hypothesis with $m-1$ and $s*n$.

Taking $s=\nil$ and $m=B(\nil)=\itp{M_Y}(\delta)$ proves that $\itp{M_Y}(\delta)$ is a uniform modulus on the inputs bounded by $\delta$.
\end{proof}

\paragraph{Spector's bar recursion}

We return to an arbitrary result type~$\sigma$.

\begin{defi}[\agdalink{BarRecursion.Spector.Spector-bar-recursor} Spector bar recursor]
\label{def:sbr-from-gbr}
Oliva and Steila~\cite[Theorem~2.4]{OS:bar} convert a general bar recursor for a monotone predicate securing $\itp{Y}$ into a Spector bar recursor for $\itp{Y}$.
Let~$\Phi$ be the System~$\T$ term implementing this conversion, with the control functional as its first argument.
Define
\[
\SBR_Y\eqdef\Phi(Y)(\GBR_Y).
\]
\end{defi}

Theorem~\ref{thm:gbr} supplies the hypotheses of the conversion,
giving the following instance of the closure
result~\cite[Corollary~3.5]{OS:bar}.

\begin{prop}[\agdalink{BarRecursion.Spector.Spector-bar-recursion-equations}]
\label{sbr-t-definability}
The interpretation of $\SBR_Y$ satisfies Spector's bar-recursion equations for $\itp{Y}$:
for every $G:\N^*\to\itp{\sigma}$, $H:\N^*\to(\N\to\itp{\sigma})\to\itp{\sigma}$, and $s:\N^*$, we have
\[
\itp{\SBR_Y}(G)(H)(s)=
\begin{cases}
G(s), & \text{if }\itp{Y}(\hat{s})<|s|,\\
H(s)(\lambda n.\itp{\SBR_Y}(G)(H)(s*n)), & \text{otherwise}.
\end{cases}
\]
\end{prop}

\paragraph{Modified bar recursion}

Take $\sigma=\iota$.
We construct a modified bar recursor in the sense of Berger and Oliva~\cite{BO:modified}.
Write $s @ \alpha$ for the sequence obtained by overwriting the first~$|s|$ entries of $\alpha$ with~$s$.
Thus $(s@\alpha)(i)=\hat{s}(i)$ for $i<|s|$ and $(s@\alpha)(i)=\alpha(i)$ otherwise.
Fix a term
\[
\mathsf{Overwrite}:
\Term(\Gamma,\iota^*\to(\iota\to\iota)\to(\iota\to\iota))
\]
with no free variables whose interpretation satisfies
$\itp{\mathsf{Overwrite}}(s)(\alpha)=s@\alpha$.
We omit its routine construction.

\begin{defi}[\agdalink{BarRecursion.Modified.modified-bar-recursor} Modified bar recursor]
\label{def:mbr-from-gbr}
Define the conversion term
\[
\Theta \eqdef \lambda y.\lambda r.\lambda H.\,
  r\bigl(\lambda s.y(\Pad(s))\bigr)
   \bigl(\lambda s.\lambda f.\,
        y(\mathsf{Overwrite}(s)(H(s)(f)))\bigr).
\]
The resulting closed term is
\[
\MBR_Y\eqdef\Theta(Y)(\GBR_Y)
  :\Term\bigl(\emptyCtx,
    (\iota^*\to(\iota\to\iota)\to\iota\to\iota)\to\iota^*\to\iota\bigr).
\]
\end{defi}

\begin{prop}[\agdalink{BarRecursion.Modified.modified-bar-recursion-equation}]
\label{mbr-t-definability}
The interpretation of $\MBR_Y$ satisfies the modified bar-recursion equation for $\itp{Y}$:
for every $H:\N^*\to(\N\to\N)\to\N\to\N$ and $s:\N^*$,
\[
\itp{\MBR_Y}(H)(s)=
\itp{Y}\bigl(s@H(s)(\lambda n.\itp{\MBR_Y}(H)(s*n))\bigr).
\]
\end{prop}

\begin{proof}
Fix $H,s$ and let $f(n)=\itp{\MBR_Y}(H)(s*n)$.
If $S_Y(s)$ holds, the base equation for $\itp{\GBR_Y}$ gives $\itp{\MBR_Y}(H)(s)=\itp{Y}(\hat{s})$.
Since $s@H(s)(f)$ extends $s$ and $S_Y$ secures $\itp{Y}$,
this is equal to $\itp{Y}(s@H(s)(f))$.
If $S_Y(s)$ fails, the step equation gives the desired equation.
\end{proof}

\section{Extensions to sums and lists}
\label{sec:gen}

We extend the Gentzen-style translation and its fundamental theorem to System~$\T$ with sum types and finite lists.
Translating the constructors and eliminators for these types requires the nucleus to act at types beyond $\iota$.

\subsection{Generalized nuclei}
\label{sec:gen-nucleus}

The proof-theoretic translation of Section~\ref{sec:neg-trans} suggests letting the nucleus act at arbitrary types of System~$\T$.

\begin{defi}[\agdalink{generalized-nucleus} Generalized nucleus]
\label{def:gen-nucleus}
A \emph{generalized nucleus} relative to an extension of System~$\T$ is a triple $(\J,\eta,\kappa)$, where $\J$ assigns to each type $\sigma$ a type $\J\sigma$ of the same system.
The operations $\eta$ and $\kappa$ are families of terms.
For each context $\Gamma$ and types $\sigma,\tau$, they have instances
\[
\eta_\sigma : \Term(\Gamma,\sigma\to\J\sigma),
\qquad
\kappa_{\sigma,\tau} :
\Term\bigl(\Gamma,(\sigma\to\J\tau)\to\J\sigma\to\J\tau\bigr).
\]
We leave the context parameter implicit and omit type indices when they can be inferred.
\end{defi}

The ground-type instances $\eta_\iota$ and $\kappa_{\iota,\iota}$ supply the nucleus operations for the translation of System~$\T$ in Definition~\ref{def:trans}.
The same translation clauses therefore apply to the original System~$\T$ terms, with translated terms in the extended system.

\subsection{System T with sums and lists}
\label{sec:sums-and-list}

In the metatheory, we write $A+B$ for the sum of types $A$ and $B$,
with injections $\inl,\inr$ and case-analysis operation $\case$.
We write $\List(A)$ for the type of finite lists over $A$,
with constructors $[]$ and $\_::\_$ for the empty list
and adjoining an element at the head, respectively, and primitive
recursor $\listrec$.

We extend System~$\T$ with binary sum types $\sigma+\tau$
and finite list types $\List(\sigma)$.
The constants for sums are
\[
\begin{aligned}
\Inl_{\sigma,\tau}
  &:\Term(\Gamma,\sigma\to\sigma+\tau),\\
\Inr_{\sigma,\tau}
  &:\Term(\Gamma,\tau\to\sigma+\tau),\\
\Case_{\sigma,\tau,\rho}
  &:\Term\bigl(\Gamma,
    (\sigma\to\rho)\to(\tau\to\rho)\to\sigma+\tau\to\rho\bigr).
\end{aligned}
\]
The constants for lists are
\[
\begin{aligned}
\Nil_\sigma
  &:\Term(\Gamma,\List(\sigma)),\\
\Cons_\sigma
  &:\Term\bigl(\Gamma,
    \sigma\to\List(\sigma)\to\List(\sigma)\bigr),\\
\LRec_{\sigma,\rho}
  &:\Term\bigl(\Gamma,
    \rho\to
    (\sigma\to\List(\sigma)\to\rho\to\rho)
    \to\List(\sigma)\to\rho\bigr).
\end{aligned}
\]
They are interpreted using sums and finite lists in the metatheory:
\[
\begin{aligned}
\itp{\sigma+\tau} &\eqdef \itp{\sigma} + \itp{\tau}, &\quad&&
\itp{\List(\sigma)} &\eqdef \List(\itp{\sigma}),\\
\itp{\Inl} &\eqdef \inl, &&&
\itp{\Nil} &\eqdef [],\\
\itp{\Inr} &\eqdef \inr, &&&
\itp{\Cons} &\eqdef \_::\_,\\
\itp{\Case} &\eqdef \case, &&&
\itp{\LRec} &\eqdef \listrec.
\end{aligned}
\]

\subsection{Translation of sums and lists}
\label{sec:extended-translation}

Gentzen's translation places a double negation in front of disjunctions (Section~\ref{sec:neg-trans}).
Following this pattern, we apply $\J$ to the sum of the translated types.
We treat lists analogously.

\begin{defi}[\agdalink{SumsAndLists.Translation.extended-translation} Translation with sums and lists]
\label{def:extended-translation}
Fix a generalized nucleus $(\J,\eta,\kappa)$ for System~$\T$ with
sums and finite lists. We call a type of the form $\iota$,
$\sigma+\tau$, or $\List(\sigma)$ a \emph{data type}.
For a data type $D$, write $\datatrans{D}$ for its translation before applying~$\J$.
The two translations are defined by mutual recursion.
\[
\begin{aligned}
\datatrans{\iota} & \eqdef \iota,\\
\datatrans{\sigma+\tau}
  &\eqdef \sigma^\J+\tau^\J,\\
\datatrans{\List(\sigma)}
  &\eqdef \List(\sigma^\J).
\end{aligned}
\qquad\qquad
\begin{aligned}
D^\J &\eqdef \J\datatrans{D},\\
(\sigma\to\tau)^\J &\eqdef \sigma^\J\to\tau^\J,\\
(\sigma\times\tau)^\J &\eqdef \sigma^\J\times\tau^\J.
\end{aligned}
\]
Contexts and variables are translated as in Definition~\ref{def:trans}.

For types $A,\rho$, we define the Kleisli extension
\[
\KE_{A,\rho} :
  \Term\bigl(\Gamma,(A\to\rho^\J)\to\J A\to\rho^\J\bigr)
\]
by structural recursion on $\rho$:
\[
\begin{aligned}
\KE_{A,D}
  &\eqdef \kappa_{A,\datatrans{D}},\\
\KE_{A,\sigma\to\tau}
  &\eqdef \lambda g.\lambda w.\lambda x.\,
    \KE_{A,\tau}(\lambda a.g(a)(x))(w),\\
\KE_{A,\sigma\times\tau}
  &\eqdef \lambda g.\lambda w.\,
    \tpair
      \bigl(\KE_{A,\sigma}(\Fst\circ g)(w)\bigr)
      \bigl(\KE_{A,\tau}(\Snd\circ g)(w)\bigr).
\end{aligned}
\]
We translate each term $t:\Term(\Gamma,\rho)$ to a term
$t^\J:\Term(\Gamma^\J,\rho^\J)$. For the original term formers,
we use the clauses of Definition~\ref{def:trans}, taking
$\KE_{\iota,\rho}$ for $\KE_\rho$. The new clauses are
\[
\begin{aligned}
(\Inl_{\sigma,\tau})^\J
  &\eqdef \eta\circ\Inl_{\sigma^\J,\tau^\J},\\
(\Inr_{\sigma,\tau})^\J
  &\eqdef \eta\circ\Inr_{\sigma^\J,\tau^\J},\\
(\Case_{\sigma,\tau,\rho})^\J
  &\eqdef \lambda f.\lambda g.\,
    \KE_{\sigma^\J+\tau^\J,\rho}
      \bigl(\Case_{\sigma^\J,\tau^\J,\rho^\J}(f)(g)\bigr),\\[.5ex]
(\Nil_\sigma)^\J
  &\eqdef \eta(\Nil_{\sigma^\J}),\\
(\Cons_\sigma)^\J
  &\eqdef \lambda a.\,\kappa(\eta\circ\Cons_{\sigma^\J}(a)),\\
(\LRec_{\sigma,\rho})^\J
  &\eqdef \lambda b.\lambda h.\,
    \KE_{\List(\sigma^\J),\rho}
      \Bigl(\LRec_{\sigma^\J,\rho^\J}(b)
        \bigl(\lambda a.\lambda l.\,h(a)(\eta(l))\bigr)\Bigr).
\end{aligned}
\]
As before, each occurrence of $\eta$, $\kappa$, or $\KE$ uses the instance for the context in which it occurs.
\end{defi}

\subsection{The extended fundamental theorem}
\label{sec:extended-fundamental}

In Section~\ref{sec:logical-relations}, the logical relation was determined by a single relation at~$\iota$. Here we choose a base relation at each data type.

\begin{defi}[\agdalink{SumsAndLists.Translation.Relations.logical-relation} Logical relation for the extended translation]
\label{def:extended-LR}
For each data type $D$, let $\R{D}$ be a relation between $\itp{D}$ and $\itp{D^\J}$.
We extend these relations to a logical relation $\R{\rho}$ at each type $\rho$ by the function and product clauses of Definition~\ref{def:LR}.
Assignments are related pointwise as in that definition.
\end{defi}

To state the conditions on the nucleus operations, we also need to relate values of type $\itp{D}$ to values of type $\itp{\datatrans{D}}$.
The latter are the inputs to the interpretation of $\eta$ and to the functions supplied to the interpretation of $\kappa$.

\begin{defi}[\agdalink{SumsAndLists.Translation.Relations.lifted-relation} Lifted relation]
\label{def:lifted-data-relation}
For each data type $D$, we define a relation $\SR{D}$ between $\itp{D}$ and $\itp{\datatrans{D}}$.
At $\iota$ it is equality.
At sums and lists it is generated inductively by
\[
\begin{gathered}
\frac{x\R{\sigma}a}
     {\inl(x)\SR{\sigma+\tau}\inl(a)}
\qquad
\frac{y\R{\tau}b}
     {\inr(y)\SR{\sigma+\tau}\inr(b)}
\qquad
\frac{}{[]\SR{\List(\sigma)}[]}
\qquad
\frac{x\R{\sigma}a \qquad l\SR{\List(\sigma)}l'}
     {(x::l)\SR{\List(\sigma)}(a::l')}.
\end{gathered}
\]
\end{defi}

Fix an arbitrary context~$\Gamma$ and assignment $\gamma:\itp{\Gamma}$.
We impose two compatibility conditions on the nucleus operations.
For each data type~$D$, the condition on $\eta$ is
\begin{equation}\label{cond:extended-eta}
  \forall x^{\itp{D}}\,a^{\itp{\datatrans{D}}}.\,
  x \SR{D} a \;\to\; x \R{D} \itp{\eta}_\gamma(a).
\end{equation}
For any data types $D,E$ and functions $f:\itp{D}\to\itp{E}$ and $g:\itp{\datatrans{D}}\to\itp{E^\J}$, the condition on $\kappa$ is
\begin{equation}\label{cond:extended-kappa}
  \bigl(\forall x\,a.\, x \SR{D} a \to f(x) \R{E} g(a)\bigr)
  \;\to\; f \R{D\to E} \itp{\kappa}_\gamma(g).
\end{equation}
The same induction as in Lemma~\ref{lem:KE-LR} extends the condition on $\kappa$ to $\KE$.

\begin{lem}[\agdalink{SumsAndLists.Translation.Relations.preservation-KE}]
\label{lem:extended-KE}
Assume condition~\eqref{cond:extended-kappa}.
For each data type $D$, type $\rho$, context $\Gamma$,
and assignment $\gamma:\itp{\Gamma}$, if $f:\itp D\to\itp\rho$ and $g:\itp{\datatrans{D}}\to\itp{\rho^\J}$ satisfy $f(x) \R{\rho} g(a)$ whenever $x \SR{D} a$, then
\[
f\R{D\to\rho}
\itp{\KE_{\datatrans{D},\rho}}_\gamma(g).
\]
\end{lem}

\begin{thm}[\agdalink{SumsAndLists.Translation.Relations.fundamental-theorem} Fundamental theorem of logical relations with sums and lists]
\label{thm:extended-LR}
Assume conditions~\eqref{cond:extended-eta}
and~\eqref{cond:extended-kappa}.
For every term $t:\Term(\Gamma,\rho)$ and assignments
$\gamma:\itp{\Gamma}$ and $\delta:\itp{\Gamma^\J}$,
\[
\gamma \R{\Gamma} \delta
\ \to \
\itp t_\gamma \R{\rho} \itp{t^\J}_\delta.
\]
In particular,
$\itp t\R\rho\itp{t^\J}$ for every closed term
$t:\Term(\emptyCtx,\rho)$.
\end{thm}

\begin{proof}
By structural induction on $t$, keeping the related assignments arbitrary.
The original System~$\T$ cases follow as in Theorem~\ref{thm:LR}, using Lemma~\ref{lem:extended-KE} for the recursor.

The injection cases follow from condition~\eqref{cond:extended-eta}.
For case analysis, $\SR{\sigma+\tau}$-related inputs select corresponding branches and supply them with related arguments.
Lemma~\ref{lem:extended-KE} then gives the required relation for $\Case$.

The empty-list case follows from condition~\eqref{cond:extended-eta}.
For $\Cons_\sigma$, fix related heads.
Adjoining these heads to $\SR{\List(\sigma)}$-related tails gives $\SR{\List(\sigma)}$-related lists.
Condition~\eqref{cond:extended-eta}, followed by condition~\eqref{cond:extended-kappa}, gives the required relation for $\Cons_\sigma$.
For $\LRec_{\sigma,\rho}$, fix related base and step arguments.
Induction on the lifted list relation relates the ordinary recursors, using condition~\eqref{cond:extended-eta} for the tail arguments.
Lemma~\ref{lem:extended-KE} then gives the result.
\end{proof}

\section{Other monadic translations}
\label{sec:variants}

We first present the Kuroda-style translation of System~$\T$ with products and then prove a fundamental theorem of logical relations for it.
With a suitable nucleus, the translation yields moduli of pointwise continuity.
We verify them using the fundamental theorem and then compare them with those obtained by the Gentzen-style construction in Section~\ref{sec:pcont}.

We also briefly recall the Kolmogorov-style translation.
For both translations, we fix a generalized nucleus $(\J,\eta,\kappa)$ relative to System~$\T$ with products, as defined in Section~\ref{sec:gen-nucleus}.
In the term definitions below, each occurrence of $\eta$ or $\kappa$ uses the instance for the context in which it occurs.

\subsection{The Kuroda-style monadic translation}
\label{sec:kuroda}

We formulate within System~$\T$ the call-by-value monadic translation underlying Powell's framework~\cite[Section~3]{powell:ccht}, using a generalized nucleus.
Its proof-theoretic analogue is the Kuroda variant that also places double negations in front of conclusions of implications~\cite[Section~6]{FO:negative}.

\paragraph{Translation}
The translation of application uses the term
\[
\begin{aligned}
\App^\bullet &
: \Term\bigl(\Gamma, \J(\sigma\to\J\tau)\to\J\sigma\to\J\tau\bigr)\\
\App^\bullet &
\eqdef \lambda f^{\J(\sigma\to\J\tau)} .\, \lambda a^{\J\sigma} .\, \kappa \bigl(\lambda g^{\sigma\to\J\tau} .\, \kappa(g)(a)\bigr)(f).
\end{aligned}
\]
For $f:\Term(\Gamma,\J(\sigma\to\J\tau))$ and $a:\Term(\Gamma,\J\sigma)$, write
\[
f \bullet a \eqdef \App^{\bullet}(f)(a).
\]
The operation $\App^{\bullet}$ uses $\kappa$ both to obtain the function from $f$ and to supply it with an argument from $a$.

\begin{defi}[\agdalink{KurodaTranslation.Kuroda-translation} Kuroda-style monadic translation]
We assign to each type $\rho$ a type~$\ku{\rho}$, defined recursively by
\[
\ku{\iota} \eqdef \iota,
\qquad
\ku{\sigma \times \tau} \eqdef \ku{\sigma} \times \ku{\tau},
\qquad
\ku{\sigma \to \tau} \eqdef \ku{\sigma} \to \J\ku{\tau}.
\]
Each term $t:\Term(\Gamma,\rho)$ is translated to a term
$\ku{t}:\Term(\ku{\Gamma},\J\ku{\rho})$,
where $\ku{\Gamma}$ replaces each declaration $x:\sigma$
in $\Gamma$ with $x:\ku{\sigma}$.
The term translation is defined inductively as follows.

\smallskip
\noindent\emph{Simply typed $\lambda$-calculus.}
\[
\ku{x}\eqdef\eta(x),
\qquad
\ku{\lambda x^\sigma.t}\eqdef\eta(\lambda x^{\ku{\sigma}}.\ku{t}),
\qquad
\ku{t(u)}\eqdef\ku{t}\bullet\ku{u}.
\]

\noindent\emph{Product constants.}
\[
\ku{\tpair} \eqdef \eta\bigl(\lambda a.\,
    \eta\bigl(\lambda b.\,
      \eta(\tpair(a)(b))\bigr)\bigr),
\qquad
\ku{\Fst} \eqdef\eta(\eta\circ\Fst),
\qquad
\ku{\Snd} \eqdef\eta(\eta\circ\Snd).
\]

\noindent\emph{Numerical constants and recursor.}
\[
\ku{\Zero}\eqdef\eta(\Zero),
\qquad
\ku{\Suc}\eqdef\eta(\eta\circ\Suc),
\qquad 
\ku{\Rec_\sigma}
\eqdef\eta\bigl(\lambda a.\,
  \eta\bigl(\lambda f.\,
    \eta(\Rec^\bullet_{\ku{\sigma}}(a)(f))\bigr)\bigr),
\]
where $\Rec^\bullet$ is defined by
\[
\begin{aligned}
\Rec^\bullet_\sigma & :
\Term\bigl(\Gamma,\,
  \sigma\to(\iota\to\J(\sigma\to\J\sigma))
  \to\iota\to\J\sigma\bigr) \\
\Rec^\bullet_\sigma & \eqdef
\lambda a.\lambda f.\,
  \Rec_{\J\sigma}(\eta(a))
    \bigl(\lambda n.\lambda r.\,f(n)\bullet r\bigr).
\end{aligned}
\]
\end{defi}

We call $\ku{\rho}$ the \emph{value type}, representing values, and $\J\ku{\rho}$ the \emph{computation type}, representing computations carrying those values.
The soundness statement uses a relation at each of these two levels.

\begin{defi}[\agdalink{KurodaTranslation.Relations.logical-relations} Logical relations for the Kuroda-style translation]
\label{def:LR-kuroda}
For each type $\rho$, choose a \emph{computation relation} $\bC{\rho}$ between $\itp{\rho}$ and $\itp{\J\ku{\rho}}$.
We define the \emph{value relation} $\bV{\rho}$ between $\itp{\rho}$ and
$\itp{\ku{\rho}}$ by structural recursion:
\[
\begin{aligned}
n \bV{\iota}m
  &\eqdef n=m,\\
u \bV{\sigma\times\tau} v
  &\eqdef
    \bigl(\fst(u) \bV{\sigma} \fst(v)\bigr)
    \land
    \bigl(\snd(u) \bV{\tau} \snd(v)\bigr),\\
f \bV{\sigma\to\tau} g
  &\eqdef
    \forall x^{\itp{\sigma}}\,a^{\itp{\ku{\sigma}}}.\,
    \bigl(x \bV{\sigma} a
      \to f(x) \bC{\tau} g(a)\bigr).
\end{aligned}
\]
Assignments $\gamma:\itp{\Gamma}$ and
$\delta:\itp{\ku{\Gamma}}$ are related pointwise:
\[
\gamma \bV{\Gamma} \delta
\ \eqdef \
\forall(x:\sigma)\in\Gamma.\,
  \gamma(x) \bV{\sigma} \delta(x).
\]
\end{defi}

Fix an arbitrary context~$\Gamma$ and assignment $\gamma:\itp{\Gamma}$.
We impose two compatibility conditions on the nucleus operations.
For each type~$\sigma$, the condition on $\eta$ is
\begin{equation}\label{cond:kuroda-eta}
  \forall x^{\itp{\sigma}}\,a^{\itp{\ku{\sigma}}}.\,
  x \bV{\sigma} a \;\to\; x \bC{\sigma} \itp{\eta}_\gamma(a).
\end{equation}
For any types $\sigma,\tau$ and functions $f:\itp{\sigma}\to\itp{\tau}$ and $g:\itp{\ku{\sigma}}\to\itp{\J\ku{\tau}}$, the condition on $\kappa$ is
\begin{equation}\label{cond:kuroda-kappa}
  \bigl(\forall x\,a.\, x \bV{\sigma} a \to f(x) \bC{\tau} g(a)\bigr)
  \;\to\;
  \bigl(\forall x\,w.\, x \bC{\sigma} w
    \to f(x) \bC{\tau} \itp{\kappa}_\gamma(g)(w)\bigr).
\end{equation}

\begin{thm}[\agdalink{KurodaTranslation.Relations.fundamental-theorem} Fundamental theorem of logical relations for the Kuroda-style translation]
\label{thm:kuroda-LR}
Assume conditions~\eqref{cond:kuroda-eta} and~\eqref{cond:kuroda-kappa}.
For every term $t:\Term(\Gamma,\rho)$ and assignments $\gamma:\itp{\Gamma}$ and $\delta:\itp{\ku{\Gamma}}$,
\[
\gamma \bV{\Gamma} \delta
\ \to \
\itp{t}_\gamma \bC{\rho} \itp{\ku{t}}_\delta.
\]
In particular,
$\itp{t}\bC{\rho}\itp{\ku{t}}$
for every closed term $t:\Term(\emptyCtx,\rho)$.
\end{thm}

\begin{proof}
By structural induction on $t$, keeping the related assignments arbitrary.

We apply condition~\eqref{cond:kuroda-eta} at each occurrence of $\eta$ introduced by the translation.
For variables, abstractions, and constants, its use at the outer $\eta$ reduces the goal to establishing
the corresponding value relation $\bV{}$.
For a variable, this follows from the relation on assignments.
For an abstraction, it follows from the induction hypothesis under assignments extended by related arguments.
For the product constants, $\Zero$, and $\Suc$, we use the corresponding clauses of $\bV{}$.

Applying condition~\eqref{cond:kuroda-kappa} twice, once for each occurrence of $\kappa$ in $\App^\bullet$, gives the following implication for any assignment $\gamma$.
\begin{equation}
\label{eq:kuroda-application}
f\bC{\sigma\to\tau}F \land x\bC{\sigma}A
\ \to \
f(x)\bC{\tau}\itp{\App^\bullet}_\gamma(F)(A).
\end{equation}
The application case follows by applying it to the two induction hypotheses.

For the recursor, fix $a \bV{\sigma} a'$ and $h\bV{\iota\to\sigma\to\sigma}h'$.
Induction on $n:\N$, using condition~\eqref{cond:kuroda-eta}
at zero and \eqref{eq:kuroda-application} at successors, gives
the following relation for any assignment $\gamma$.
\[
\natrec_{\itp{\sigma}}(a)(h)(n)
\bC{\sigma}
\itp{\Rec^\bullet_{\ku{\sigma}}}_\gamma(a')(h')(n).
\]
Since $\bV{\iota}$ is equality, this gives $\natrec_{\itp{\sigma}}(a)(h) \bV{\iota \to \sigma} \itp{\Rec^\bullet_{\ku{\sigma}}}_\gamma(a')(h')$.
Applying condition~\eqref{cond:kuroda-eta} at the two inner occurrences of $\eta$ in $\ku{\Rec_\sigma}$ gives the required value relation for $\Rec_\sigma$.
\end{proof}

\paragraph{Pointwise continuity}

We now use the Kuroda-style translation to extract moduli of pointwise continuity, for comparison with the Gentzen-style construction in Section~\ref{sec:pcont}.
A computation takes an incoming modulus and an input sequence and returns a value together with a modulus.

\begin{defi}[\agdalink{PointwiseContinuityViaKuroda.pointwise-continuity-nucleus} Pointwise-continuity nucleus for the Kuroda-style translation]
\label{def:kuroda-cont-nucleus}
The generalized nucleus $(\J,\eta,\kappa)$ is given by
\[
\begin{aligned}
\J\sigma
  &\eqdef \iota\to(\iota\to\iota)\to(\sigma\times\iota),\\[1ex]
\eta_\sigma
  &:\Term(\Gamma,\sigma\to\J\sigma)\\
\eta_\sigma
  &\eqdef \lambda x.\lambda n.\lambda\alpha.\,\tpair(x)(n),\\[1ex]
\kappa_{\sigma,\tau}
  &:\Term\bigl(\Gamma,(\sigma\to\J\tau)\to\J\sigma\to\J\tau\bigr)\\
\kappa_{\sigma,\tau}
  &\eqdef \lambda g.\lambda w.\lambda n.\lambda\alpha.\,
    \tpair(\Fst(u_g))
      \bigl(\Max(\Max(\Snd(u_g))(\Snd(u_w)))(n)\bigr),
\end{aligned}
\]
where $u_w$ abbreviates $w(n)(\alpha)$ and $u_g$ abbreviates $g(\Fst(u_w))(n)(\alpha)$ in the body of $\kappa_{\sigma,\tau}$.
\end{defi}

The operation $\eta$ returns its argument with the incoming modulus.
In the definition of $\kappa$, both computations receive the same incoming modulus~$n$.
The value returned by $w$ determines which computation to select from $g$.
The outgoing modulus is the maximum of~$n$ and the moduli returned by the two computations.

For $w:\itp{\J\sigma}$, we write $w(n)(\alpha)=\pair{\rV_w^n(\alpha)}{\rM_w^n(\alpha)}$.

\begin{defi}[\agdalink{PointwiseContinuityViaKuroda.value-and-computation-relations} Value and computation relations for pointwise continuity]
\label{def:kuroda-cont-relations}
For every $\alpha:\N\to\N$ and $n:\N$, we define $\bVV{\rho}$ and $\bCC{\rho}$ simultaneously by induction on~$\rho$.
For $\bVV{\rho}$ we use the value-relation clauses of Definition~\ref{def:LR-kuroda}, with $\bCC{\tau}$ in the function clause.
For $x:\itp{\rho}$ and $w:\itp{\J\ku{\rho}}$, define
\[
x \bCC{\rho} w
\ \eqdef \
\bigl( x \bVV{\rho} \rV_w^n(\alpha) \bigr)
\land
\bigl( n \leq \rM_w^n(\alpha) \bigr)
\land
\bigl( \forall\beta.\,
    \alpha =_{\rM_w^n(\alpha)} \beta
    \to \rV_w^n(\beta)=\rV_w^n(\alpha)\bigr).
\]
\end{defi}

The computation relation says that, at $\alpha$, $w$ returns a value related to $x$ and a modulus of continuity for $\rV_w^n$ at $\alpha$.
The returned modulus is at least the incoming modulus~$n$.

\begin{lem}[\agdalink{PointwiseContinuityViaKuroda.preservation}]
\label{lem:kuroda-cont-preservation}
For every $\alpha:\N\to\N$ and $n:\N$, the nucleus of
Definition~\ref{def:kuroda-cont-nucleus} satisfies
conditions~\eqref{cond:kuroda-eta} and~\eqref{cond:kuroda-kappa}
for the relations of Definition~\ref{def:kuroda-cont-relations}.
\end{lem}

\begin{proof}
Fix $\alpha$ and $n$.
If $x \bVV{\sigma} a$, then $\itp{\eta}(a)$ returns $a$ at every input with modulus $n$, giving $x \bCC{\sigma} \itp{\eta}(a)$. So the condition on $\eta$ is satisfied.

For the condition on $\kappa$, let $f$ and $g$ satisfy its premise and suppose $x\bCC{\sigma}w$.
Put $a=\rV_w^n(\alpha)$ and $u=\itp{\kappa}(g)(w)$.
Since $x\bVV{\sigma}a$, the premise gives $f(x)\bCC{\tau}g(a)$, and hence $f(x)\bVV{\tau}\rV_u^n(\alpha)$.
The modulus $\rM_u^n(\alpha)$ is the maximum of $n$, the modulus returned by $w$ for selecting $a$, and the modulus returned by $g(a)$ for its value.
It therefore gives the required modulus for $\rV_u^n$ at $\alpha$.
\end{proof}

As in Section~\ref{sec:pcont}, we apply the translated functional to a generic element $\Omega$.
Here the original argument has type $\iota\to\iota$, so translated application expects an argument of computation type $\J\ku{\iota\to\iota}=\J(\iota\to\J\iota)$.

\begin{defi}[\agdalink{PointwiseContinuityViaKuroda.generic-element} Generic element for Kuroda-style pointwise continuity]
\label{def:kuroda-cont-omega}
We define the term $\Omega :\Term(\Gamma,\J(\iota\to\J\iota))$ by $\Omega \eqdef \eta(q)$ with
\[
\begin{aligned}
q &:\Term(\Gamma,\iota\to\J\iota)\\
q &\eqdef \lambda i.\lambda n.\lambda\beta.\,
  \tpair(\beta(i))(\Max(\Suc(i))(n)).
\end{aligned}
\]
\end{defi}

\begin{lem}[\agdalink{PointwiseContinuityViaKuroda.genericity-of-\%CE\%A9}]
\label{lem:kuroda-cont-genericity}
For every $\alpha:\N\to\N$ and $n:\N$, we have
\[
\alpha\bVV{\iota\to\iota}\itp{q}
\qquad\text{and}\qquad
\alpha\bCC{\iota\to\iota}\itp{\Omega}.
\]
\end{lem}

\begin{proof}
For each $i:\N$, the computation $\itp{q}(i)$ at $n$ and $\alpha$ returns $\alpha(i)$ with modulus $\max\{i+1,n\}$.
Agreement on that many entries implies agreement at $i$, giving the first claim.
The second follows from the condition on $\eta$.
\end{proof}

\begin{cor}[\agdalink{PointwiseContinuityViaKuroda.pointwise-continuity}]
\label{cor:kuroda-continuity}
For every closed term $f:\Term(\emptyCtx,(\iota\to\iota)\to\iota)$, the interpretation
of the term
\[
M_f \eqdef \Snd\circ\bigl((\ku{f}\bullet\Omega)(\Zero)\bigr)
\]
is a modulus of continuity of $\itp{f}$.
\end{cor}

\begin{proof}
Let $w=\itp{\ku{f}\bullet\Omega}$.
For each $\alpha:\N\to\N$, take the incoming modulus to be zero.
Lemma~\ref{lem:kuroda-cont-preservation} allows us to apply Theorem~\ref{thm:kuroda-LR} to $f$.
Together with Lemma~\ref{lem:kuroda-cont-genericity} and \eqref{eq:kuroda-application}, this gives $\itp{f}(\alpha)\bCCz{\iota}w$.
Thus $\rV_w^0(\alpha)=\itp{f}(\alpha)$, and $\rM_w^0(\alpha)$ is a modulus of continuity for $\rV_w^0$ at $\alpha$.
Since the value equality holds at every input, $\rM_w^0(\alpha)=\itp{M_f}(\alpha)$ is also a modulus of continuity for $\itp{f}$ at $\alpha$.
\end{proof}

\paragraph{Evaluation order and comparison}
In the Kuroda-style translation, the two uses of~$\kappa$ in translated application combine the moduli from computing the function, computing its argument, and applying the resulting function to that argument.
The Gentzen-style translation uses ordinary application, so an argument contributes a modulus only when the translated function uses it.

\begin{exa}[\agdalink{ComparingExtractedModuli.comparison-of-extracted-moduli} Comparing extracted moduli]
\label{ex:continuity-comparison}
Consider the closed term
\[
t \eqdef \lambda\alpha.\,
  \Rec_\iota(\alpha(\Zero))(\lambda n.\lambda r.\,\Zero)(\Suc(\Zero)).
\]
Its interpretation is constantly zero, as the single recursion step discards the recursive result.

In the Kuroda-style translation, applying the translated recursor to its base argument already accounts for the query at index zero.
This is the only index queried, so the extracted modulus is constantly one.
In the Gentzen-style translation (Section~\ref{sec:pcont}), the step discards the translated base argument, so the extracted modulus is constantly zero.
\end{exa}

\subsection{The Kolmogorov-style monadic translation}
\label{sec:kolmogorov}

For any generalized nucleus on System~$\T$, we define a call-by-name monadic translation within the same system, following the CPS translation of Barthe and Uustalu~\cite{BU:CPS} and its monadic generalization~\cite{uustalu:monad}.
The placement of $\J$ mirrors the placement of double negations in Kolmogorov's negative translation.

For application, we define
\[
\begin{aligned}
\App^\diamond &:
  \Term\bigl(\Gamma,\J(\sigma\to\J\tau)\to\sigma\to\J\tau\bigr)\\
\App^\diamond &\eqdef
  \lambda f.\lambda a.\,\kappa(\lambda g.g(a))(f)
\end{aligned}
\]
and write $f\diamond a\eqdef\App^\diamond(f)(a)$.

\begin{defi}[\agdalink{KolmogorovTranslation.Kolmogorov-translation} Kolmogorov-style monadic translation]
A type $\rho$ translates to $\J\ko{\rho}$, where
\[
\ko{\iota}\eqdef\iota,
\qquad
\ko{\sigma\times\tau}\eqdef\J\ko{\sigma}\times\J\ko{\tau},
\qquad
\ko{\sigma\to\tau}\eqdef\J\ko{\sigma}\to\J\ko{\tau}.
\]
A term $t:\Term(\Gamma,\rho)$ translates to $\ko{t}:\Term(\ko{\Gamma},\J\ko{\rho})$, where $\ko{\Gamma}$ replaces each declaration $x:\sigma$ in $\Gamma$ with $x:\J\ko{\sigma}$.
The term translation is defined inductively as follows.

\smallskip
\noindent\emph{Simply typed $\lambda$-calculus.}
\[
\ko{x}\eqdef x,
\qquad
\ko{\lambda x^\sigma.t}
  \eqdef \eta\bigl(\lambda x^{\J\ko{\sigma}}.\ko{t}\bigr),
\qquad
\ko{t(u)}\eqdef\ko{t}\diamond\ko{u}.
\]

\noindent\emph{Product constants.}
\[
\ko{\tpair}
  \eqdef \eta\bigl(\lambda a.\,
    \eta\bigl(\lambda b.\,\eta(\tpair(a)(b))\bigr)\bigr),
\qquad
\ko{\Fst}\eqdef\eta(\kappa(\Fst)),
\qquad
\ko{\Snd}\eqdef\eta(\kappa(\Snd)).
\]

\noindent\emph{Numerical constants and recursor.}
\[
\ko{\Zero}\eqdef\eta(\Zero),
\qquad
\ko{\Suc}\eqdef\eta\bigl(\kappa(\lambda n.\,\eta(\Suc(n)))\bigr).
\]
For the recursor, we define
\[
\begin{aligned}
\Rec^\diamond_\sigma &:
  \Term\bigl(\Gamma,\,
    \sigma\to(\J\iota\to\J(\J\sigma\to\J\sigma))
    \to\iota\to\J\sigma\bigr)\\
\Rec^\diamond_\sigma &\eqdef
  \lambda a.\lambda f.\,
  \Rec_{\J\sigma}(\eta(a))
    \bigl(\lambda n.\lambda r.\,f(\eta(n))\diamond r\bigr).
\end{aligned}
\]
The translation of the recursor is
\[
\ko{\Rec_\sigma}
\eqdef
\eta\Bigl(\kappa\Bigl(\lambda a.\,
  \eta\bigl(\kappa\bigl(\lambda f.\,
    \eta\bigl(\kappa(\Rec^\diamond_{\ko{\sigma}}(a)(f))\bigr)
  \bigr)\bigr)
\Bigr)\Bigr).
\]
\end{defi}

\begin{rem}[\agdalink{KolmogorovTranslation.translated-application}]
For terms $t$ of type $\sigma\to\tau$ and $u$ of type $\sigma$, the three translations handle function application as follows.
\begin{center}
\setlength{\abovetopsep}{1ex}
\setlength{\belowbottomsep}{1ex}
\begin{tabular}{@{}lccc@{}}
\toprule
Translation
  & \shortstack{Translated\\function type}
  & \shortstack{Translated\\argument type}
  & \shortstack{Translated\\application} \\
\midrule
Gentzen
  & $\sigma^\J\to\tau^\J$
  & $\sigma^\J$
  & $t^\J(u^\J)$ \\
Kuroda
  & $\J(\ku{\sigma}\to\J\ku{\tau})$
  & $\J\ku{\sigma}$
  & $\ku{t}\bullet\ku{u}$ \\
Kolmogorov \qquad
  & $\J(\J\ko{\sigma}\to\J\ko{\tau})$
  & $\J\ko{\sigma}$
  & $\ko{t}\diamond\ko{u}$ \\
\bottomrule
\end{tabular}
\end{center}
The Gentzen-style translation applies the translated function directly.
The Kuroda-style translation uses $\kappa$ to obtain the function and again to obtain an argument of type $\ku{\sigma}$.
The Kolmogorov-style translation uses $\kappa$ to obtain the function, which accepts the translated argument at its computation type $\J\ko{\sigma}$.
\end{rem}

\section{Conclusions and future work}
\label{sec:conclusion}

The Gentzen-style monadic translation provides a common framework for extracting computational witnesses from terms of System~$\T$.
The choice of nucleus determines the translation.
The fundamental theorem relates the interpretation of each term to that of its translation whenever the nucleus operations satisfy the compatibility conditions for the chosen logical relation.
The applications illustrate the range of witnesses obtainable in this way.
The treatment of sums and lists suggests that the approach can extend to richer languages.
The comparison with the Kuroda-style translation shows how different treatments of function application can yield different witnesses for the same property.

Several directions merit further investigation.
The first is to extend the Gentzen-style translation and its fundamental theorem to strictly positive inductive and coinductive types.
Our treatment of natural numbers, sums, and lists provides a starting point for this extension.

The second is to relate and combine nuclei \emph{within a fixed translation}.
Relations between nuclei could allow us to compare the witnesses extracted by their translations, while suitable combinations could capture information from several interpretations at once.
The aim is to develop a uniform fundamental theorem supporting these constructions and comparisons.

The third is to relate nuclei \emph{for different translations} that establish the same property.
For pointwise continuity, the Gentzen- and Kuroda-style translations use different nuclei and logical relations.
Can we systematically derive the nucleus and logical relation for one construction from those for the other?

\bibliographystyle{alphaurl}
\begingroup
\appto\UrlBreaks{\do\-}
\bibliography{gmt-bib}
\endgroup

\end{document}